\documentclass{article}
\usepackage{kr}

\usepackage{times}
\usepackage{soul}
\usepackage{url}
\usepackage[hidelinks]{hyperref}
\usepackage[utf8]{inputenc}
\usepackage[small]{caption}
\usepackage{graphicx}
\usepackage{amsmath}
\usepackage{amsthm}
\usepackage{booktabs}
\usepackage{algorithm}
\usepackage{algorithmic}
\usepackage{pgfplots}
\usepgfplotslibrary{external}
\pgfplotsset{compat=1.18}
\usepackage{pgfplotstable}

\newcommand{\hideablePlot}[1]{#1}

\usepackage{makecell}
\usepackage{adjustbox}
\usepackage{tabularx}
\usepackage{tikz}
\usepackage[many]{tcolorbox}

\newtheorem{example}{Example}
\newtheorem{theorem}{Theorem}

\newtheorem{lemma}{Lemma}

\usepackage{amssymb}
\newcommand{\anonymize}[2]{#1}
\newcommand{\anonAuth}[1]{\anonymize{#1}{Anonymous Author}}
\newcommand{\anonAff}[1]{\anonymize{#1}{Anonymous Affiliation}}
\newcommand{\anonEmails}[1]{\anonymize{#1}{Anonymous Emails}}

\title{Identifying and Explaining (Non-)Equivalence of First-Order Logic Formulas}

\newcommand{\narrowerAuthors}{\hspace*{-1.5mm}}
\author{\anonAuth{Fabian Vehlken}$^1$\narrowerAuthors\and
  \anonAuth{Thomas Zeume}$^1$\narrowerAuthors\and
  \anonAuth{Emilio Carrasco Bustamante}$ ^2 $\narrowerAuthors\and
  \anonAuth{Ma\"elle Corn\'ely}$^3$\narrowerAuthors\and
  \anonAuth{Lukas Pradel}$^1$\\
\affiliations
$^1$\anonAff{Ruhr University Bochum}\\
$^2$\anonAff{TU Dortmund University}\\
$^3$\anonAff{Universit\'{e} Paris-Saclay, ENS Paris-Saclay}\\
\emails
\anonEmails{\{fabian.vehlken, thomas.zeume\}@rub.de}}

\newif\ifcomments
\newif\ifchanges
\commentsfalse\changesfalse
\commentstrue\changestrue 

\newif\ifwithappendix
\withappendixfalse
\withappendixtrue

\newcommand{\insertDependingOnAppendix}[2]{\ifwithappendix #1 \else #2 \fi}

\usepackage{fvcols}
\usepackage{xspace}
\usepackage[capitalise]{cleveref}
\usepackage{subcaption}
\usepackage{enumitem}

\newcommand{\df}{\ensuremath{\mathrel{\smash{\stackrel{\scriptscriptstyle{
    \text{def}}}{=}}}} \;}

\newcommand  {\myclass} [1]  {\ensuremath{\textsf{\upshape #1}}}

\newcommand{\StaClass}[1]{\myclass{#1}\xspace}

\newcommand  {\algorithmicProblem} [1] {\normalfont{\textsc{#1}}\xspace}

\newcommand{\problemIndent}{\hspace{-4mm}}
\newcommand{\algorithmicProblemDescription}[4][\linewidth]{
    \vspace{1mm}
    \def\Name{#2}
    \def\Input{#3}
    \def\Question{#4}
      \problemIndent
      \setlength{\tabcolsep}{1mm}
      \begin{tabularx}{#1}{rX}\textit{Problem:}&\algorithmicProblem{\Name} \\
     \textit{Input:}&\Input \\
     \textit{Question:}&\Question
     \end{tabularx}\vspace{1mm}
    }

\newcommand{\CQ}[1][]{\StaClass{CQ}}
\newcommand{\UCQ}[1][]{\StaClass{UCQ}}
\newcommand{\CQneg}[1][]{\StaClass{CQ\ensuremath{^{\mneg}}}}
\newcommand{\UCQneg}[1][]{\StaClass{UCQ\ensuremath{^{\mneg}}}}

\theoremstyle{plain}

\theoremstyle{definition}

\newtheorem*{question*}{Question}
\newtheorem*{openquestion*}{Open question}

\providecommand {\calA}      {{\mathcal A}\xspace}
\providecommand {\calB}      {{\mathcal B}\xspace}

\providecommand {\calG}      {{\mathcal G}\xspace}

\providecommand {\calP}      {{\mathcal P}\xspace}

\providecommand {\calS}      {{\mathcal S}\xspace}

\ifcomments
\newcommand{\commentbox}[1]{\noindent\framebox{\parbox{0.98\linewidth}{#1}}}

\newcommand{\acomment}[2]{\ \\ \fbox{\parbox{0.98\linewidth}{{\sc #1}: #2}}}
\newcommand{\mcomment}[2]{{\color{blue}(#1)}\footnote{#1: #2}} \else
\newcommand{\commentbox}[1]{}
\newcommand{\mcomment}[2]{}
\newcommand{\acomment}[2]{}
\fi

\ifchanges

\setul{}{0.2mm}
\setstcolor{red}

\else

\fi

\newcommand{\constraints}{\ensuremath{\Gamma}}
\newcommand{\Mod}{\ensuremath{\mathrm{Mod}}}
\providecommand{\mathcolor}[2]{\begingroup\color{#1}#2\endgroup}
\newcommand{\tpl}[1]{\overline{#1}}

\newcommand{\tableindent}{\hspace{0.5em}}

\newcommand{\fiiill}{&&&&&&&&&&&&&}
\newcommand{\defaultColumnHeaders}{\# & \% & \#dst & \%dst}
\begin{document}

\maketitle

\begin{abstract}
  First-order logic is the basis for many knowledge representation formalisms and methods. Providing technological support for learning to write first-order formulas for natural language specifications requires methods to test formulas for (non-)equivalence and to provide explanations for non-equivalence. We propose such methods based on both theoretical insights and existing tools, implement them, and report on experiments testing their effectiveness on a large educational data set with $> 100{,}000$ pairs of first-order formulas.
\end{abstract}

\section{Introduction}\label{section:introduction}
First-order logic is central for multiple areas of computer science. It is the basis of formal languages for representing knowledge and reasoning, of specification languages for formal verification, and of query languages for databases, among others. In knowledge representation and reasoning it is used in many subdisciplines, including argumentation (e.g. to formalize arguments and their relationships), belief change, logic programming, and reasoning about knowledge and belief. Also, many other knowledge representation formalisms are based on (fragments of) first-order logic, e.g.\ description logics, modal logics, and non-monotonic logics.

Basics of first-order logic -- modelling, normal forms, and reasoning -- are therefore often included in mandatory courses in Bachelor Computer Science programmes; for instance, guidelines of the German Association for Computer Science (GI) recommend to include such basics in computer science study programmes \cite{GI2016}. Figure \ref{figure:fo-modelling-exercise} sketches a typical exercise of such a course which students learn to solve by (1) designing a suitable first-order vocabulary (i.e., which relation, function, and constant symbols to use), (2) describing the scenario by first-order formulas, (3) transforming formulas into a suitable normal form (e.g.\ prenex or Skolem form), and (4) inferring knowledge using an inference mechanism (e.g.\ resolution).

In this paper we continue a line of research aimed at supporting educational tasks (1) -- (4) for learning first-order logic in educational support systems. While basic support for all four tasks is provided by state-of-the-art systems such as the Iltis system \cite{SchmellenkampVZ24,KneiselVZ25}, not much work has been done on providing advanced feedback, even though individual, timely feedback is essential for learning success. In this work we focus on advanced feedback mechanisms for writing first-order formulas, i.e.\ for educational task (2).

Writing first-order formulas is a challenge for students and prone to many errors \cite{SchmellenkampSZ25}. Besides stating that a student attempt is incorrect and providing counter examples, feedback ideally pinpoints conceptual mistakes. For instance, 
when asked to write a first-order formula for the natural language statement
\begin{itemize}
 \item[] ``\textit{System libraries depend only on system libraries.}''
\end{itemize}
with solution
\begin{itemize}
 \item[] \(\forall x \forall y \big((S(x) \land D(x,y)) \to S(y)\big),\)
\end{itemize}
we envision that students receive individual conceptual feedback depending on their attempt, for instance:
\begin{itemize}
\item[(a)] Hints at quantification mistakes:
        \begin{itemize}
         \item Student: $ \forall x \exists y \big((S(x) \land D(x,y)) \to  S(y)\big) $
         \item Feedback: It might help to double check the quantifier structure of your formula.
        \end{itemize}

		 \item[(b)]  Hints at guarding mistakes:
        \begin{itemize}
         \item Student: $ \forall x \forall y (D(x,y) \to S(y)) $
         \item Feedback: It seems that you need to ensure that the variable $x$ ranges only over some domain elements.
        \end{itemize}

    \item[(c)]  Hints at missing parts or symbols of a formula:
        \begin{itemize}
         \item Student: $ \forall x \forall y (S(x) \to S(y)) $
         \item Feedback: The statement includes a dependency between libraries, but your formula does not use the relation symbol $D$ that models such dependencies.
        \end{itemize}
\end{itemize}

Algorithmically providing feedback faces two obstacles.
First, providing feedback for first-order formulas is an inherently hard algorithmic problem. Already determining the correctness of student attempts requires checking logical equivalence of first-order formulas, often modulo a background theory for knowledge implicit in the given scenario. While this algorithmic problem is undecidable in general, one approach is to use a standard reduction to the satisfiability problem for first-order logic and to apply a theorem prover such as Vampire \cite{KovacsV13,BartekBCHHHKRRRSSV25}. A natural question is whether this suffices and also whether small, meaningful counter examples can be computed this way for non-equivalent formulas.

\begin{itemize}
    \item RQ1: Can the correctness of student attempts be determined in practice for first-order logic modelling tasks?
    \item RQ2: In case of an incorrect attempt, can a counter example witnessing incorrectness be provided?
\end{itemize}

Second, no approaches for providing higher level explanations (such as explanations of types (a) -- (c) from above) for explaining why an attempt is wrong or, in other words, why two first-order formulas are non-equivalent, have been explored in prior work. In particular, no algorithmic methods for finding such explanations exist so far.

\begin{itemize}
  \item RQ3: In case of an incorrect attempt, how can high-level explanations be computed in practice?
\end{itemize}

\subsection*{Contributions}

We address research questions RQ1, RQ2, and RQ3 by (1) designing algorithmic methods for explaining non-equivalence of first-order formulas and (2) evaluating equivalence testing, finding counter models, and the methods for explaining non-equivalence from (1) on large educational data sets with $> 100{,}000$ pairs of first-order formulas.

Our explanations for non-equivalence of two formulas $\psi$ and $\varphi$ come in two flavours. A \emph{blocker} is a syntactic property of one of them that prevents equivalence; examples are different numbers of free variables or missing vocabulary symbols. A \emph{bugfixing modification} is a small modification of one of the formulas that makes it equivalent to the other; examples are bugfixes that modify a quantifier, a guarding attempt or small careless mistakes. For computing explanations, we exploit several techniques:
\begin{itemize}
 \item We use consequences of characterizations of Beth definability \cite[Theorem 6.6.4]{Hodges93} to computationally discover vocabulary symbols that must be used for expressing some property.
 \item We use specifically designed profiling information for quantifier patterns of atoms and (attempted) guards to compute candidates for bugfixing modifications, whose correctness is then tested by applying them and testing equivalence to the resulting formula.
 \item We use a na\"{\i}ve Erd\"os-R\'enyi-style random model generator to generate counter examples to supplement built-in finite model building algorithms of Vampire.
\end{itemize}

We evaluate all methods -- equivalence testing, finding counter models, methods for explaining non-equivalence -- on an educational data set collected anonymously with an educational support system in mandatory Bachelor-level courses at several large universities. At the time of data collection, only feedback indicating whether the input was correct, possibly augmented by a counter example, was provided (with the possibility of receiving feedback that the system was not able to determine the correctness of the input).

In total, the data set contains 125,970 pairs of formulas (distinct: 37,159), of which 27,655 pairs are equivalent (distinct: 3,019) and 98,315 pairs are non-equivalent (distinct: 34,140). Each pair consists of a first-order formula modelling a natural language statement and a formula by a student attempting to model the same statement. In total, there are pairs for 62 statements spread over 14 scenarios, 11 of them with a background theory of implicit knowledge.

For an evaluation summary we refer to \cref{table:summary-all-exercises}. We answer research questions RQ1 and RQ2 affirmatively. For all 125,970 pairs of formulas, equivalence is determined by the Vampire theorem prover. For 99.14\% of non-equivalent pairs of formulas a counter example is found either by Vampire's finite model builder (97.16\%, exclusively 19.1\%) or our na\"{\i}ve random model generator (80.04\%, exclusively 1.98\%). Towards answering research question RQ3, our computational strategies for computing explanations yield at least one explanation for 52.15\% of all non-equivalent pairs. Strategies based on missing symbols, quantifiers, and guarding each explain $> 20\%$ of non-equivalent pairs.

These results lay the foundation for conceptual feedback for first-order modelling in educational support systems. They also quantitatively confirm conceptual challenges encountered by students in first-order modelling tasks, including the use of quantifiers and guarding, and thus can be the basis for both fine-grained analysis and intervention design from a CS education research perspective.

\begin{figure}[t]
	\begin{minipage}{0.48\textwidth}
		\begin{tcolorbox}[    tikznode boxed title,
    enhanced,
    arc=0mm,
    interior style={white},
    attach boxed title to top left= {yshift=-\tcboxedtitleheight/2, xshift=3mm},
    fonttitle=\bfseries,
    colbacktitle=white,coltitle=black,
    boxed title style={size=normal,colframe=white,boxrule=0pt},
    title={Exercise: From modelling to inference},
    left=4pt,right=4pt,
    top=0pt, bottom=0pt,
		]
			\vspace{2mm}
\noindent\small
                        Julia is investigating a collection of software packages. Each software package has exactly one maintainer, and software packages may depend on other software packages.
                        She has identified the following dependencies:

			\begin{enumerate}[leftmargin=6mm]
                            \item[(a)] No software package is both a program library and a system library.
                            \item[(b)] System libraries depend only on system libraries.
                            \item[(c)] Program libraries may only depend on system libraries which are not maintained by the same maintainer.
                            \item[(d)] Every software package that must be explicitly installed by the user depends on at least one program library directly.
			\end{enumerate}

                        \vspace{-1.3mm}
			\noindent She concludes that there is no system library that must be explicitly installed by the user.
\noindent Can you confirm her conclusion using methods you learned for first-order logic?
		\end{tcolorbox}
\end{minipage}
 \caption{A first-order logic modelling exercise in which students have to choose a first-order vocabulary, write first-order formulas for natural language specifications in this vocabulary, and infer knowledge with an inference mechanism such as resolution.}\label{figure:fo-modelling-exercise}\end{figure}

\subsection*{Related work}
How to explain student mistakes and how to provide adequate feedback has also been explored for educational modelling tasks in other subfields of formal foundations of computer science including propositional logic \cite{SchmellenkampLZ23}, temporal logics \cite{GreenmanSNK23}, regular languages \cite{DantoniKAGVH2015}, and context-free languages \cite{SchmellenkampZAKSS2025}.

\subsection*{Outline}
After formalizing the setting and recalling methods for determining (non-)equivalence of formulas and finding counter examples in \cref{section:setting}, we discuss methods for explaining non-equivalence in \cref{section:methods}. All methods are evaluated in \cref{section:evaluation}. We discuss impact and future work in \cref{section:summary}.
\insertDependingOnAppendix{Appendix~\ref{section:appendix-exercises} provides further details on the evaluation data set, including a sample exercise, while Appendix~\ref{section:appendix-evaluation-tables} presents detailed evaluation data.}{An extended version of this paper with an appendix containing details on the evaluations and sample exercises is available on arXiv \cite{VehlkenZeumeCBCP2026extd}.} 
\section{Setting, Equivalence, Counter Examples}\label{section:setting}
In educational tasks that ask to write first-order formulas for natural language statements, typically students are given a natural language description of a scenario together with the statement to be modelled. For writing the formula, a first-order vocabulary $\sigma$ with a description of the intended meaning of its symbols is provided by instructors or designed by students in preliminary educational tasks. An instructor specifies a first-order formula $\psi$ over $\sigma$ that correctly models the statement as well as additional first-order constraints $\Gamma$ over $\sigma$ that model constraints implicit in the scenario description. The first-order constraints $\Gamma$ form a \emph{background theory}. We denote its class of models by $\Mod(\Gamma) \df \{\calS \mid \calS \models \Gamma\}$. For checking whether a first-order formula $\varphi$ over $\sigma$ written by a student correctly describes the statement to be modelled, one checks whether $\psi$ and $\varphi$ agree for all $\sigma$-structures $\calS \in \Mod(\Gamma)$.

For determining correctness and certifying incorrectness, the following algorithmic problem needs to be~solved:

\algorithmicProblemDescription{FO-Equivalence-Modulo-Theory}
        {
            First-order formulas $\psi$, $\varphi$ and a set $\Gamma$ of first-order sentences over $\sigma$.
        }
        {
            Are $\psi$ and $\varphi$ equivalent modulo $\Gamma$ (denoted by $\psi \equiv_{\Gamma} \varphi$), that is, does $\calS \models \psi$ iff $\calS \models \varphi$ hold for all $\calS \in \Mod(\Gamma)$? If not, which  $\calS \in \Mod(\Gamma)$ is a counter example?
        }

We explain existing algorithmic approaches for this problem, as they are the basis for more advanced explanations (see \cref{section:methods}) and because their performance on educational data sets is evaluated later on (see \cref{section:evaluation}).

The decision version of the problem above can be easily reduced to first-order (un-)satisfiability as follows:
\begin{align*}
    \psi\equiv_{\Gamma} \varphi &\iff \bigwedge\constraints \to (\psi\leftrightarrow \varphi) \text{ is a tautology}\\
                                          &\iff \bigwedge\constraints \land \lnot (\psi\leftrightarrow \varphi) \text{ is unsatisfiable.}
\end{align*}
Thus, it can be solved by invoking a theorem prover such as Vampire \cite{KovacsV13,BartekBCHHHKRRRSSV25}.

For providing counter examples for non-equivalent formulas, one can use Vampire's finite model building implementation \cite{RSV16}. Since initial experiments indicated that this does not produce counter examples for all pairs judged to be non-equivalent by Vampire, we also use a na\"{\i}ve random model generator approach. We compare both approaches in Section \ref{section:evaluation}.

The na\"{\i}ve random model approach uses an Erd\H{o}s-R\'enyi-style construction (see e.g. \cite{Libkin04}). Given a universe size $ n $ and a probability $ p $, construct a structure $ \mathcal{S} $ with universe $ A = \{1, \dots, n\} $. For each $ k $-ary relation symbol $ R \in \sigma $, each tuple in $ A^k $ will be added to $ R^\mathcal{S} $ independently with probability $ p $. For $ k $-ary function symbols $ f \in \sigma $, the image under $ f^\mathcal{S} $ of each tuple in $ A^k $ is chosen randomly, each element in $ A $ having probability $ \frac{1}{n} $ to be picked.
The universe sizes, probability, and a number of models to be generated for each universe size are parameters.

To find counter examples, one generates random structures in ascending size and checks if they are a model of $ \bigwedge\constraints \land \lnot (\psi\leftrightarrow \varphi) $. If so, it is a counter example certifying \mbox{$ \psi\not \equiv_{\Gamma} \varphi $}. By checking $ \bigwedge\constraints \land \psi \land \lnot \varphi$ and $ \bigwedge\constraints \land \lnot \psi \land \varphi $ separately, additionally feedback as to whether a student attempt $\varphi$ is too restrictive or too permissive can be provided.

Our implementation allows to randomly generate structures satisfying $\Gamma$ in a preprocessing step, and only to subsequently test against $\psi$ and $\varphi$. This is helpful in cases where $\Gamma$ severely restricts permissible models. In such cases, this preprocessing speeds up the counter example search.
 
\section{Explanations for Non-Equivalence}\label{section:methods}

\begin{table*}[t]
    \small
    \centering
     \setlength\tabcolsep{1mm}
     \renewcommand{\arraystretch}{1.2}
    \begin{tabular}{l||c|c|p{4.5cm}}
        \toprule
        \textbf{Strategy} & \multicolumn{2}{c|}{\textbf{Example}} & {\hspace*{4mm}}\textbf{Explanation} \\
        & Solution formula $\psi$  & Student attempt $\varphi$ & \\
        \midrule
        \multicolumn{4}{c}{\textbf{Strategies for Symbols}}\\
        \midrule
        S-1 Missing symbol names & $\forall x\, (Q(x) \rightarrow P(x))$ & $\forall x\, P(x)$ & Relation symbol $ Q $ does not occur in $\varphi$, but is required \\
        S-2 Permuted arguments of a symbol & $ \forall x \exists y\, R(\mathcolor{blue}{x},\mathcolor{blue}{y}) $ & $ \forall x \exists y\, R(\mathcolor{red}{y},\mathcolor{red}{x}) $ & Wrong quantification pattern due to permutation \\
        S-3 Different relation symbol names & $ \exists x\, \mathcolor{blue}{P}(x) $ & $ \exists x\, \mathcolor{red}{Q}(x) $ & Wrong relation symbol used\\
        S-4 Different terms & $ \exists x \, P(\mathcolor{blue}{f(}x\mathcolor{blue}{)}) $ & $ \exists x\, P(x) $ & Terms in $P(\cdot)$ differ\\
        \midrule
        \multicolumn{4}{c}{\textbf{Strategies for Quantifiers}}\\
        \midrule
        Q-1 Different quantifier prefixes &  $ \forall x \mathcolor{blue}{\exists y} (P(x) \to G(x,y)) $ & $ \forall x \mathcolor{red}{\forall y} (P(x) \to G(x,y)) $ & Wrong quantifier prefix \\
        Q-2 Different quantification patterns & $ \forall x (S(x) \to \exists y \mathcolor{blue}{\forall z}\, R(x,y,z)) $ & $ \forall x \exists y \mathcolor{red}{\exists z} (S(x) \to R(x,y,z)) $ & Wrong quantification pattern for $ R(x,y,z) $\\
        Q-3 Different free variables & $\forall x\, P(x)$ & $ P(x)$ & $ x $ is free only in $\varphi$ \\
        \midrule
        \multicolumn{4}{c}{\textbf{Strategies for Guarding}}\\
        \midrule
        G-1 Different guards & $\forall x\, (\mathcolor{blue}{P(x) \rightarrow} Q(x))$ & $\forall x\, Q(x)$ & Missing universal guard for $x$ \\
        & $\exists x\, Q(x)$ & $\exists x\, (\mathcolor{red}{P(x) \land} Q(x))$ & Superfluous existential guard for $x$ \\
        G-2 Wrong guarding operators & $ \forall x\, (P(x) \mathcolor{blue}{\to} Q(x)) $ & $ \forall x\, (P(x) \mathcolor{red}{\land} Q(x)) $ & $\wedge$ is the wrong guard operator for universally quantified $x$ \\
        \midrule
        \multicolumn{4}{c}{\textbf{Strategies for Boolean Operators$^*$}}\\
        \midrule
        B-1 Different negation patterns &  $ \forall x\, P(x) $ & $ \forall x\, \mathcolor{red}{\lnot} P(x) $ & Wrong negation prefix for $P(x)$\\
        B-2 Swapping Implication &$ \forall x\, (\mathcolor{blue}{P(x)} \rightarrow \mathcolor{blue}{Q(x)}) $ & $ \forall x\, (\mathcolor{red}{Q(x)} \rightarrow \mathcolor{red}{P(x)}) $ & Implication in wrong direction\\
    \end{tabular}
      \caption{Overview of strategies for explaining non-equivalence of first-order formulas. Only most successful strategies for Boolean operators are listed ($^*$).}
\label{table:strategy-overview}
\end{table*}

Our conceptual explanations for non-equivalence are motivated from our educational application scenario. A recent interview study on errors and misconceptions in first-order logic modelling \cite{SchmellenkampSZ25} confirmed the common perception among instructors that when writing first-order formulas, students struggle with the use of (i) quantifiers, (ii) guarding, and (iii) free variables. Additionally, students often (iv) forget to model parts of statements, or (v) do careless syntactical mistakes including wrong use of Boolean operators. Our goal is to explain non-equivalences of first-order formulas that occur because of mistakes due to (i) -- (v).

In this section we always assume non-equivalence of the two first-order formulas $\psi$ and $\varphi$. Our explanations for non-equivalence are \emph{directional} in the sense that they try to explain why $\varphi$ is not equivalent to $\psi$ from $\varphi$'s perspective. This allows to derive appropriate feedback for student attempts from computed explanations.

Roughly speaking, our explanations come in two flavours: \emph{blockers} and \emph{bugfixing modifications}.   A \emph{blocker} is a syntactic property of $\varphi$ that no formula equivalent to $\psi$ (modulo $ \constraints $) can have. An example is that $\psi$ uses a relation symbol $R$ and no equivalent formula without $R$ exists; thus if $\varphi$ does not use $R$, this explains non-equivalence of $\varphi$ and $\psi$. A \emph{bugfixing modification} is a small modification of $\varphi$ such that applying it to $\varphi$ leads to a formula equivalent to $\psi$.

We next describe algorithmic strategies to discover explanations for non-equivalence. A challenge is efficiency: for finding bugfixing explanations, our algorithms enumerate promising candidate modifications for $\varphi$ and test whether they transform $\varphi$ into $\varphi'$  such that $\varphi' \equiv_{\Gamma} \psi$. Thus, each candidate requires a computationally expensive equivalence test (e.g.\ via Vampire). Our strategies are designed to minimize equivalence tests and to find good candidates.

\cref{table:strategy-overview} provides an overview of our computational strategies as well as examples for their application.

\subsection{Preprocessing}
In a preprocessing step, we compute information that profiles (i) use of quantifier in $\varphi$ and $\psi$, (ii) use of guards in $\varphi$ and $\psi$, and (iii) necessary vocabulary symbols in $\psi$.

\subsubsection{Profiling the Use of Quantifiers.}
We precompute quantifier profiles of atoms occurring in $\psi$ and $\varphi$. 

A first-order formula is in \emph{prenex normal form}, if it is of the form $Q_1 x_1 \ldots Q_k x_k\, \mu$ for quantifiers $Q_i\in \{\exists, \forall\}$ and quantifier-free formula $\mu$. We call $Q_1 x_1 \ldots Q_k x_k$ \emph{quantifier prefix} and $Q_1Q_2\ldots Q_k$ \emph{quantifier prefix type}. Denote by $\varphi^\neg$ the \emph{negation normal form} obtained from $\varphi$ by ``moving all negations inwards'' using De Morgan's law.

The \emph{quantifier prefix of an atom} $\calA \df R(x_1, \ldots, x_k)$ in formula $\varphi$ captures how variables in the atom are quantified and in which order. Formally, it is the sequence $Q_1 y_1 \ldots Q_\ell y_\ell$ where $Q_i$ are quantifiers and $y_i$ variables with
\begin{itemize}
 \item $\{y_1, \ldots, y_\ell\} \cup \{z_1, \ldots, z_m\} = \{x_1, \ldots, x_k\}$, where $z_i$ are free variables of $\varphi$, and
 \item $y_1, \ldots, y_\ell$ are bound by $Q_1 y_1, \ldots, Q_\ell y_\ell$ in this order in the syntax tree of $\varphi^\neg$ (if $y_i$ is quantified multiple times, the sequence contains the quantification closest to $\calA$).
\end{itemize}
The \emph{quantifier prefix type} of $\calA$ captures how and in which order positions of $\calA$ are quantified. Formally, it is the sequence $Q_1 P_1, \ldots, Q_\ell P_\ell$ where $P_i \df \{j \mid x_j = y_i\}$.
The address of an atom $\calA$ in a formula $\varphi$ is a path $\pi_\calA$ from the root of the syntax tree of $\varphi$ to $\calA$. The \emph{profile} $\text{profile}(\varphi, \calA, \pi_\calA)$ of an atom $\calA \df R(x_1, \ldots, x_k)$ with address $\pi_\calA$ in $\varphi$ is the tuple $(R, v, q)$ where the valence $v$ is $v = +$ if no negation is in front of $\calA$ in $\varphi^\neg$ and $v = -$ otherwise, and $q$ is the quantifier prefix type of $\calA$ in $\varphi$. The \emph{profile} $\text{profile}(\varphi)$ of a formula $\varphi$ is the set $\{\text{profile}(\varphi, \calA, \pi_\calA) \mid \text{$\calA$ is an atom with address $\pi_\calA$ in $\varphi$}\}$.

\begin{example}
 Consider the first-order formula
 \[\varphi \df \exists y \big(S(y) \wedge \forall x \neg \forall y (T(x,y,x) \vee S(y))\big).\] The quantifier prefix of the atom $T(x,y,x)$ is $\forall x \exists y$, and its quantifier prefix type is $\forall\{1,3\} \exists\{2\}$. The profile of $\varphi$ is
 \[\{(S, +, \exists\{1\}), (S, -, \exists\{1\}), (T, -, \forall\{1,3\} \exists\{2\})\}.\]
\end{example}

For the sake of readability, we introduced a simplified notion of profiles that only talks about variables. In our implementation, profiles contain information about terms occurring in atoms and about the quantifier structure of variables occurring in those terms.

In a preprocessing step for $\varphi$ and $\psi$, our procedure computes $\text{profile}(\psi)$, $\text{profile}(\varphi)$, and $\text{prefix}(\psi)$ and $\text{prefix}(\varphi)$ (if $\varphi$, $\psi$ are in prenex form).

\subsubsection{Profiling the Use of Guards.}

The use of a variable $z$ in an atom $\calA(\tpl u)$ within a formula $\psi$  is \emph{guarded} by an atom $\calG(\tpl v)$, if $z = y_\ell$ in the quantifier prefix $Q_1 y_1 \ldots Q_k y_\ell$ of $\calG(\tpl u)$ and there is a subformula of the formula $\psi$ of the form

\begin{itemize}
 \item $\calG \rightarrow\rho(\tpl w)$ and $\calA(\tpl u)$ occurs in $\rho(\tpl w)$ if $Q_\ell = \forall$, and
 \item $\calG \wedge \rho(\tpl w)$ and $\calA(\tpl u)$ occurs in $\rho(\tpl w)$ if $Q_\ell = \exists$.
\end{itemize}
It is \emph{wrongly guarded} if the conditions are replaced by
\begin{itemize}
 \item $\calG \wedge \rho(\tpl w)$ and $\calA(\tpl u)$ occurs in $\rho(\tpl w)$ if $Q_\ell = \forall$, and
 \item $\calG \rightarrow \rho(\tpl w)$ and $\calA(\tpl u)$ occurs in $\rho(\tpl w)$ if $Q_\ell = \exists$.
\end{itemize}
As an example, $x$ in the atom $R(x,y,z)$ within $\forall y \exists x(S(x) \wedge \exists z R(x,y,z))$ is guarded by $S(x)$, while it is wrongly guarded within  $\forall y \exists x(S(x) \rightarrow \exists z R(x,y,z))$.

Let $\text{guards}(\psi)$ be the set that contains all pairs $(\textsc{profile}(\calG), \textsc{profile}(\calA))$ of profiles of atoms $\calG$ and $\calA$ such that there is a variable $z$ in $\calA$ guarded by $\calG$ within $\psi$. Similarly, let $\text{wrong-guards}(\psi)$ contain all such pairs where $z$ is wrongly guarded.

In a preprocessing step for $\varphi$ and $\psi$, our procedure computes $\text{guards}(\psi)$, $\text{guards}(\varphi)$, $\text{wrong-guards}(\psi)$, and $\text{wrong-guards}(\varphi)$.

\subsubsection{Profiling Necessary Symbols.}
If $\psi$ and $\varphi$ are non-equivalent and $\varphi$ does not use a symbol $R$ used by $\psi$, one possible explanation for the non-equivalence is that $R$ is necessary for expressing the property expressed by $\psi$.

Formally, let $\psi$ be a first-order formula over $\sigma$ that expresses a property $\calP$ modulo a background theory $\Gamma$. A symbol $R \in \sigma$ is \emph{necessary} to express $\calP$ modulo $\Gamma$, if there is no $\varphi$ over $\sigma \setminus \{R\}$ that is equivalent to $\psi$.

Whether a symbol is necessary can be inferred using the projective Beth definability theorem, see for instance \cite{tenCateC25}. The following theorem rephrases the formulation from \cite[Theorem 6.6.4]{Hodges93}:

\begin{theorem}\label{thm:missingSymbol}
    Let $ \sigma, \sigma' $ be vocabularies with $ \sigma \subseteq \sigma'$, let $\psi(\bar{x})$ be a first-order formula over $\sigma'$, 
and let $ \Gamma $ be a set of first-order sentences over $ \sigma' $ such that $ \Mod(\Gamma) \neq \emptyset $.
    Then the following are equivalent:
    \begin{enumerate}
        \item[(i)] There exist two $ \sigma' $-structures $ \mathcal{A}, \mathcal{B} \models \Gamma $ such that
            \begin{itemize}
                \item $ \mathcal{A}|\sigma = \mathcal{B}|\sigma $ and
\item for some tuple $ \bar{a} $, $ \mathcal{A} \models \psi(\bar{a}) \iff \mathcal{B} \not\models \psi(\bar{a}) $.
            \end{itemize}\item[(ii)] There exists no first-order formula $ \varphi(\bar{x}) $ over $\sigma$ such that $ \Gamma \models \forall \bar{x} (\psi(\bar{x}) \leftrightarrow \varphi(\bar{x})) $.
\end{enumerate}\end{theorem}

\cref{thm:missingSymbol} immediately helps in our application.
For determining whether a symbol \(R \in \sigma'\) is necessary for expressing the property defined by a solution formula \(\psi\) modulo a background theory \(\Gamma\), we apply \cref{thm:missingSymbol} with \(\sigma = \sigma' \setminus \{R\}\).
To this end, we employ Padoa's method (see e.g.\ \cite{tenCateC25}) to test whether structures \(\mathcal{A}\) and \(\mathcal{B}\) and a tuple \(\bar{a}\) as in part (i) of the theorem exist.
If so, the symbol \(R\) is necessary in any student attempt formula \(\varphi\) equivalent to \(\psi\) modulo \(\Gamma\).

We reduce the testing of part (i) of \cref{thm:missingSymbol}, i.e.\ whether structures $\calA$ and $\calB$ exist, to the satisfiability problem for first-order logic and apply an automated theorem prover like Vampire. For finite $ \Gamma $, as in our application, such structures exist if and only if the following existential second-order (ESO) formula $ \chi $ is satisfiable:

\newcommand{\symbolReplacement}[1]{\ensuremath{[\tpl T/ \tpl T#1,\; \tpl g/ \tpl g#1]}}
\vspace{1mm}
    $\exists \tpl R \, \exists \tpl f \exists \tpl T' \exists \tpl g' \exists \tpl T'' \exists \tpl g''  \exists \bar{x}$\\ \vspace*{1mm}
    \hspace*{1cm}$\Big(\big(\psi_{\symbolReplacement{'}}(\bar{x}) \leftrightarrow \lnot \psi_{\symbolReplacement{''}}(\bar{x})\big)$\\ \vspace{1mm}
    \hspace*{1.8cm}$\land~\bigwedge\Gamma_{\symbolReplacement{'}} \land \bigwedge\Gamma_{\symbolReplacement{''}} \Big)$

where
\begin{itemize}
    \item $\exists \tpl R \, \exists \tpl f$ quantifies a structure over $\sigma$,
    \item $\exists \tpl T' \exists \tpl g'$ and  $\exists \tpl T'' \exists \tpl g''$ quantify structures over $\sigma' \setminus \sigma$, and
        \item $ \psi_{[S/S']} $ denotes the formula obtained from $ \psi $ by replacing every occurrence of $ S $ by $ S' $ (analogously for $ \Gamma $).
\end{itemize}

The ESO formula $ \chi $ is satisfiable if and only if the following first-order logic formula is satisfiable:

    $\exists \bar{x}\,\Big(\big(\psi_{\symbolReplacement{'}}(\bar{x}) \leftrightarrow \lnot\psi_{\symbolReplacement{''}}(\bar{x})\big)$\\ \vspace{1mm}
    \hphantom{$\exists \bar{x}\,$}\hspace*{1cm}$\land~\bigwedge\Gamma_{\symbolReplacement{'}} \land \bigwedge\Gamma_{\symbolReplacement{''}} \Big)$

We note that in \cite{Hodges93}, all first-order vocabularies contain equality. Therefore, in \cref{thm:missingSymbol}, if $ \sigma' $ contains equality, then $ \sigma $ also contains equality.
However, in our application, we also want to determine whether the equality symbol is necessary to express the property expressed by $\psi$ on $\Mod(\Gamma)$. To accommodate this, we rephrase the question of whether equality is necessary as a question of whether a congruence relation, i.e.\ an equivalence relation compatible with all other symbols, is necessary. This can be tested using \cref{thm:missingSymbol}.
We rely on the fact that factoring a structure with a congruence by that congruence collapses this relation into the equality relation.

We reuse notation from \cite[Exercise 11.6.11]{EbbinghausFT1994} and state the following lemma. For a first-order formula $\psi$ (with equality) over $\sigma$, we denote by $\psi^*$ the first-order formula over $\sigma \cup \{E\}$, for a fresh binary relation symbol $E$, obtained by replacing all atoms $t_1 = t_2$ in $\psi$ by $E(t_1, t_2)$, where $t_1, t_2$ are terms. Furthermore, let $\delta_E$ be a first-order formula axiomatizing that $E$ is a congruence with respect to all symbols in $\sigma$.  For a set $\Gamma$ of first-order sentences we denote $\Gamma^* \df \{\varphi^* \mid \varphi \in \Gamma\}$.

\begin{lemma}\label{lemma:equality-to-equivalence}
    Let $ \psi $ be a first-order formula and $ \Gamma $ a finite set of first-order sentences over a vocabulary $ \sigma $ (with equality).
    Then the following are equivalent:
  \begin{itemize}
   \item[(i)] There exists a formula $ \varphi $ over $ \sigma $ (without equality) such that $\psi \equiv_{\Gamma} \varphi$.
   \item[(ii)] There exists a formula $ \mu $ over $ \sigma $ (without $ E $) such that $\psi^* \equiv_{\Gamma^* \cup \{\delta_E\}} \mu $.
  \end{itemize}
\end{lemma}

Thus, whether equality is needed in $\psi$ reduces to testing whether the symbol $E$ is necessary in $\psi^*$. 

\begin{proof}[Proof (of Lemma \ref{lemma:equality-to-equivalence})]

   We use the following fact, see \cite{EbbinghausFT1994}. For all formulas $\psi$, all models $\calS$, all binary relations $E^\calS$ with $(\calS, E^\calS) \models \delta_E$ it holds that $(\calS, E^\calS) \models \psi^*$ iff $\calS / E^\calS \models \psi$, where $\calS / E^\calS$ denotes the factor structure of $\calS$ by $E^\calS$. 

   Using this fact, one can easily show that $\psi \equiv_{\Gamma} \varphi$ iff $\psi^* \equiv_{\Gamma^* \cup \{\delta_E\}} \varphi^*$ for any two formulas $\psi$, $\varphi$ and set $\Gamma$ of sentences.
   Direction (i) $ \implies $ (ii) follows, since $ \varphi $ does not use equality by assumption, and therefore $ \varphi^* $ does not use~$ E $. Thus, $ \varphi^* = \varphi $ is a valid choice for $ \mu $.

   For (ii) $ \implies $ (i), 
   we observe that in the proof of the projective Beth definability theorem in \cite{tenCateC25}, the constructed formula that explicitly defines the symbol in question (in this case, $ E $) in terms of $\sigma$ does not use equality, if the background theory does not use equality, which $ \Gamma^* \cup \{\delta_E\} $ does not. Therefore, if $ \mu $ as in (ii) exists, there also exists a formula $ \mu' $ equivalent to $ \mu $ modulo $ \Gamma^* \cup \{\delta_E\} $ that uses neither equality nor $ E $. Thus, $ \mu' $ is a valid choice for $ \varphi $ in (i).

    The lemma statement follows.
\end{proof}

In a preprocessing step for $\psi$, all necessary symbols $\textsc{necessary}(\psi)$ of $\psi$ are computed relying on the described procedure. Note that this requires a first-order satisfiability test for each symbol and is therefore expensive, but this only has to be done once for each solution formula $\psi$.

\subsection{Strategies for Explaining Non-equivalence}
 We now explain computational strategies for explaining non-equivalence; most using the precomputed information.

\subsubsection{Strategies for Symbols.} The following strategies identify missing symbols and
candidates for fixing non-equivalence due to permuted arguments or wrong names of symbols.
\begin{description}
\item[S-1 Missing symbol names] If a vocabulary symbol (relation, function, or constant) is necessary in $\psi$ but does not occur in $\varphi$, then this is a blocker explaining non-equivalence.

\item[S-2 Permuted arguments of a symbol] If only $\varphi$ has an atom $\calA_\varphi$ with profile $(R, v, q_\varphi)$ and only $\psi$ has an atom with profile $(R, v, q_\psi)$, and permuting arguments of $\calA_\varphi$ changes $q_\varphi$ into $q_\psi$, then modify $\varphi$ by permuting the arguments. If successful (i.e., if this leads to equivalence to $\psi$), then this modification explains non-equivalence.

\item[S-3 Different relation symbol names] If only $\varphi$ has an atom $\calA_\varphi$ with profile $(R_\varphi, v, q)$ and only $\psi$ has an atom  with profile $(R_\psi, v, q)$, then modify $\varphi$ by replacing the symbol $R_\varphi$ in $\calA_\varphi$ by $R_\psi$.  If successful (i.e., if this leads to equivalence to $\psi$), then this modification explains non-equivalence.

\item[S-4 Different terms] Similarly to the previous case but for terms (relying on extended profiles that incorporate information for terms).
\end{description}

\subsubsection{Strategies for Quantifiers.}
The following strategies identify different quantification patterns and different sets of free variables.
\begin{description}
 \item[Q-1 \textbf{Different quantifier prefixes}] If both $ \psi $ and $ \varphi $ are in prenex normal form and their quantifier prefixes differ, then modify $\varphi$ by replacing its prefix by the prefix of $\psi$ if this does not introduce additional free variables. If successful (i.e., if this leads to equivalence to $\psi$), then this modification explains non-equivalence.
 \item[Q-2 \textbf{Different quantification patterns}] If only $\varphi$ has an atom $\calA_\varphi$ with profile $(R, v, q_\varphi)$ and only $\psi$ has an atom with profile $(R, v, q_\psi)$, then modify $\varphi$ by replacing\footnotemark[1] the quantifiers corresponding to $q_\varphi$ in $\varphi$ by quantifiers corresponding to $q_\psi$ (with suitable renaming of variables). If successful (i.e., if this leads to equivalence to $\psi$), then this modification explains non-equivalence.
 \item[Q-3 \textbf{Different free variables}] If the set of free variables of $\psi$ and $\varphi$ differs,  then this is a blocker explaining non-equivalence.
\end{description}

\subsubsection{Strategies for Guarding.}
The following strategies identify missing and superfluous guards in $\varphi$ as well as wrong guards.
\begin{description}
 \item[G-1 Different guards] If $\varphi$ contains an atom $\calA_\varphi$ with an unguarded argument position $i$ and $\psi$ contains an atom $\calA_\psi$ with the same profile but position $i$ guarded by a guard $\calG_\psi$, then modify $\varphi$ by adding a guard $\calG_\varphi$ for position $i$ of $\calA_\varphi$ (if possible). Symmetrically for $\psi$ and $\varphi$ (the modification then removes guards from $ \varphi $). If successful (i.e., if this leads to equivalence to $\psi$), then this modification explains non-equivalence.
\item [G-2 Wrong guarding operators] If $\varphi$ contains a wrongly guarded atom, then modify $\varphi$ by replacing $\rightarrow$ by $\wedge$ in the guard or vice versa.

\end{description}

\subsubsection{Strategies for Boolean Operators.}
Strategies for discovering explanations for non-equivalence for Boolean operators based on modifications have been studied extensively in \cite{SchmellenkampLZ23} and can be re-used for first-order logic. We only present a strategy for negations that is specific for first-order logic, and one exemplary Boolean strategy concerning implication (chosen because it explains many non-equivalences in our data set).

\begin{description}
 \item[B-1 \textbf{Different negation patterns}] If only $\varphi$ has an atom $\calA_\varphi$ with profile  $(R, v_\varphi, q)$  and only $\psi$ has an atom with profile $(R, v_\psi, q)$ where $v_\varphi = +$ iff $v_\psi = -$ (i.e., $R$ is positive in exactly one of the two formulas) then modify\footnotemark[1] negations on the path of the atom with profile $(R, v_\varphi, q)$ in $\varphi$ according to the corresponding path of the atom  with profile $(R, v_\psi, q)$ in $\psi$. If successful (i.e. if this leads to equivalence to $\psi$), then this modification explains non-equivalence.
\item [B-2 \textbf{Swapping Implication}] Modify $ \varphi $ by swapping the direction of an implication (if one exists). If successful (i.e. if this leads to equivalence to $\psi$), then this modification explains non-equivalence.

\end{description}
\footnotetext[1]{This rough description is made precise in the implementation.}

\section{Evaluation}\label{section:evaluation}

\renewcommand{\fiiill}{&&&&}
\begin{table*}
    \centering
    \setlength{\tabcolsep}{10pt}
    \def\arraystretch{1.3}
    \begin{tabular}{l|| rrrr}
        \toprule
        Evaluation summary                                                               & \multicolumn{4}{c}{overall attempts}                             \\
                                                                                         & \defaultColumnHeaders                                            \\
        \midrule
                all attempts                                                                     & 125970                               & 100.0\% & 37159 & 100.0\% \\
        equivalent attempts                                                              & 27655                                & 21.95\% & 3019  & 8.12\%  \\
        non-equivalent attempts                                                          & 98315                                & 78.05\% & 34140 & 91.88\% \\
        \midrule
        \multicolumn{5}{c}{\textbf{Equivalence}$ ^* $}                                                                                                      \\
        \midrule
        proven by Vampire (\texttt{vampire}, \texttt{casc}, \texttt{casc\_sat} combined) & 27655 & 100.0\% & 3019 & 100.0\% \\
        \midrule
        \multicolumn{5}{c}{\textbf{Non-Equivalence}$ ^{**} $}                                                                                               \\
        \midrule
        proven by Vampire (\texttt{vampire}, \texttt{casc}, \texttt{casc\_sat} combined) & 98315                                & 100.0\% & 34140 & 100.0\% \\
        counter model found                                                              & 97465 & 99.14\% & 33929 & 99.38\% \\
        \tableindent{} via Vampire (finite model building)                               & 95518 & 97.16\% & 33238 & 97.36\% \\
        \tableindent{} via Random models                                                 & 78691 & 80.04\% & 24168 & 70.79\% \\
        \tableindent{} exclusively via Vampire (finite model building)                   & 18774 & 19.10\% & 9761 & 28.59\% \\
        \tableindent{} exclusively via Random models                                     & 1947 & 1.98\% & 691 & 2.02\% \\
        \midrule
        \multicolumn{5}{c}{\textbf{Explanations}$ ^{**} $}                                                                                                  \\
        \midrule
        Explanation by strategy \fiiill                                                                                                                     \\
        \tableindent{} S-1 Missing symbol names                                          & 25500                                & 25.94\% & 8052  & 23.58\% \\
        \tableindent{} S-2 Permuted arguments of a symbol                                & 429                                  & 0.44\%  & 126   & 0.37\%  \\
        \tableindent{} S-3 Different relation symbol names                               & 594                                  & 0.60\%   & 123   & 0.36\%  \\
        \tableindent{} S-4 Different terms                                               & 1016                                 & 1.03\%  & 116   & 0.34\%  \\
        \tableindent{} Q-1 Different quantifier prefixes                                 & 1603                                 & 1.63\%  & 357   & 1.05\%  \\
        \tableindent{} Q-2 Different quantification patterns                             & 1047                                 & 1.06\%  & 300   & 0.88\%  \\
        \tableindent{} Q-3 Different free variables                                      & 10171                                & 10.34\% & 5375  & 15.74\% \\
        \tableindent{} G-1 Different guards                                              & 11407                                & 11.60\%  & 2038  & 5.97\%  \\
        \tableindent{} G-2 Wrong guarding operators                                      & 7445                                 & 7.57\%  & 1019  & 2.98\%  \\
        \tableindent{} Q-1 + G-1: Different quantifier prefixes and guards               & 10469                                & 10.65\% & 1439  & 4.21\%  \\
        \tableindent{} B-1 Different negation patterns                                   & 582                                  & 0.59\%  & 110   & 0.32\%  \\
        \tableindent{} B-2 Swapped implications                                          & 554                                  & 0.56\%  & 31    & 0.09\%  \\
        $\ge 1$ strategy                                                                 & 51269                                & 52.15\% & 15610 & 45.72\% \\
        \bottomrule
    \end{tabular}
    \caption{Summary of the evaluation.\\
        \setlength\tabcolsep{2pt}
        \begin{tabular}{r l}
        $ ^* $ & Percentages relative to equivalent attempts. \\
        $ ^{**} $ & Percentages relative to non-equivalent attempts.
        \end{tabular}}\label{table:summary-all-exercises}
\end{table*}

In this section, we evaluate effectiveness and efficiency of determining (non-)equivalence, finding counter examples, and computing explanations on a large educational data set. We address the research questions stated in the introduction:

\begin{itemize}
    \item RQ1: Can the correctness of student attempts be determined in practice for first-order logic modelling tasks?
    \item RQ2: In case of an incorrect attempt, can a counterexample witnessing incorrectness be provided?
    \item RQ3: In case of an incorrect attempt, how can high-level explanations be computed in practice?
\end{itemize}

Our data set was collected via an educational support system over four semesters in courses on ``Logic in Computer Science'' at two universities with $> 300$ students in each course iteration. The students were tasked with modelling statements given in natural language with first-order logic formulas over a given vocabulary, for which an instructor had provided a solution formula and (if applicable) a background theory. The data has been collected for 14 scenarios with in total 62 different statements to be modelled (see \insertDependingOnAppendix{Appendix~\ref{section:appendix-exercises}}{extended version \cite{VehlkenZeumeCBCP2026extd}} for an overview as well as an example scenario).

The data set consists of 125,970 tuples $(\text{id}, \Gamma, \psi, \varphi)$, containing an \text{id}, background theory $\Gamma$, solution formula $\psi$, and student attempt $\varphi$. In our evaluation we measure performance with respect to (i) all such tuples, and (ii) tuples with distinct solution attempts.  The number of distinct solution attempts is 37,159. The data set is summarized in the first lines of Table \ref{table:summary-all-exercises}.

One reason for the high number of attempts, and in particular incorrect attempts, is likely that the system provides feedback  after every attempt and students have an unlimited number of attempts per exercise.  Because of the large number of attempts, we used timeouts in our evaluations.
Unless stated otherwise, the experiments were run with a timeout of 20 seconds.

\subsection{RQ1: Determining Correctness of Student Attempts}

\subsubsection{Setting.}
For answering whether correctness of student attempts can be determined in practice, we evaluate (1) equivalence for all pairs of formulas $\psi$ and $\varphi$ modulo $\Gamma$ with the theorem prover Vampire, and (2) measure how fast \mbox{(non-)}equivalence can be tested.

This and all other experiments in this paper were run on a small server (16x 2.4 GHz CPU cores, 48 GB RAM). We use version $ 4.8 $ of Vampire and the three modes $ \mathtt{vampire} $, $ \mathtt{casc} $, and $ \mathtt{casc\_sat} $ \cite{BartekBCHHHKRRRSSV25}. Preliminary experiments indicated that no single mode is optimal for all our formula pairs. Therefore, we start a \emph{race} between the three different modes by running them in parallel in independent processes. As soon as one returns a decisive answer (indicating equivalence or non-equivalence), this answer is returned and all other processes are terminated.

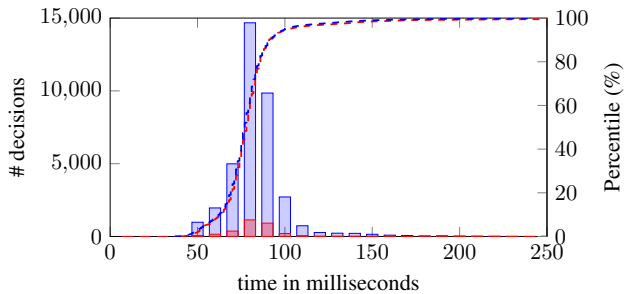
\begin{figure}[t!]
    \begin{tikzpicture}[scale=0.85]
        \begin{axis}[
            name=hist,
            width=\linewidth,
            height=5cm,
            xlabel={time in milliseconds},
            ylabel={\# decisions},
            xticklabel={\pgfmathparse{\tick*10}\pgfmathprintnumber{\pgfmathresult}},
            xmin=0,
            xmax=12,
            ymax=32000,
            enlargelimits=false,
            ytick pos=left,
            scaled y ticks=false,
            yticklabel style={
                /pgf/number format/fixed,
                /pgf/number format/1000 sep={,\!}
            },
        ]

        \addplot[
            ybar,
            draw=blue,
            fill=blue,
            bar width=5pt,
            fill opacity=.2,
            bar shift=-3pt
        ] table [
            x index=0,
            y index=1 ]
        {plot-data/decisiveTimeBuckets.txt};

        \addplot[
            ybar,
            draw=red,
            fill=red,
            bar width=5pt,
            fill opacity=.3,
            bar shift=3pt
        ] table [
            x index=0,
            y index=1 ]
        {plot-data/equivTimeBuckets.txt};

        \end{axis}

        \begin{axis}[
            at={(hist.south west)},
            anchor=south west,
            width=\linewidth,
            height=5cm,
xmin=0,
            xmax=12,
            axis x line=none,
            axis y line*=right,
            ylabel={Percentile (\%)},
            ymin=0,
            ymax=100,
ytick pos=right,
            minor y tick num=0,  
        ]

        \addplot[
            red,
            thick,
            dashed
        ] table [
            x index=0,
            y index=1,
            mark=none,
            each nth point=10
        ] {plot-data/equivPercentilesx10Vampire.txt};
        
        \addplot[
            blue,
            thick,
            dashed
        ] table [
            x index=0,
            y index=1,
            mark=none,
            each nth point=100
        ] {plot-data/decisivePercentilesx10Vampire.txt};
        \end{axis}
    \end{tikzpicture}
    \caption{Time in milliseconds Vampire needed to determine whether pairs of formulas are equivalent or non-equivalent (blue) and time needed for equivalent pairs only (red). The blue and red lines show the percentiles for decisions made. Only distinct pairs of formulas are considered here.
    \\ Note: The bars are shifted slightly to the sides for visual clarity.
    }
    \label{figure:equivalence:time}
\end{figure}

\subsubsection{Results.} For all of the 125,970 formula pairs, \mbox{(non-)} equivalence was determined within the 20 second timeout (see Table \ref{table:summary-all-exercises}). It turns out that for all except 26 pairs (21 distinct pairs), \mbox{(non-)}equivalence could be decided within 120ms, also see Figure \ref{figure:equivalence:time}.

In summary, testing correctness of solution attempts in first-order modelling tasks is practically feasible in educational scenarios.

\subsection{RQ2: Finding Counter Examples}

\subsubsection{Setting.}
For answering whether counter examples witnessing non-equivalence of formulas can be found in practice, we evaluate (1) for which ratio of pairs Vampire and a random model generator can find counter example models, (2) how fast, and (3) how large the generated counter models are (as a proxy for how useful counter examples are for explaining non-equivalence). We note that both methods only build finite models, so it is possible that counter examples for some of the formula pairs cannot be found this way.

The random model generator was parameterized to generate models of size 1 to 10, the probability of picking a tuple was $ 0.5 $ (cf. \cref{section:setting}). For each universe size $ m $, up to $ m \cdot 1000 $ models were generated and tested. Vampire was instructed to use its finite model builder via the flag \texttt{--saturation\_algorithm fmb}.

\begin{figure}[t!]
\begin{tikzpicture}[scale=0.85]
    \begin{axis}[
        ybar,
        xmin=0.2,
        xmax=7.7,
        ymin=0,
        ymax=16000,
        xtick={1,...,7},
        scaled y ticks=false,
        yticklabel style={
            /pgf/number format/fixed,
            /pgf/number format/1000 sep={,\!}
        },
        width=\linewidth,
        height=5cm,
        ylabel={\# models},
        xlabel={model sizes},
        bar width=5pt,
    ]
\addplot+[
        ybar,
        draw=blue,
        fill=blue,
        bar width=5pt,
        fill opacity=.2,
        bar shift=-3pt
    ] coordinates {
        (1,15375)
(2,11602)
        (3,5042)
        (4,1193)
        (5,18)
        (6,6)
        (7,2)
        (8,0)
        (9,0)
        (10,0)
        (11,0)
        (12,0)
        (13,0)
    };
    \addplot+[
        ybar,
        draw=red,
        fill=red,
        bar width=5pt,
        fill opacity=.2,
        bar shift=3pt
    ] coordinates {
(1,15662)
        (2,8000)
        (3,481)
        (4,25)
        (5,0)
        (6,0)
        (7,0)
        (8,0)
        (9,0)
        (10,0)
        (11,0)
        (12,0)
        (13,0)
    };

\end{axis}
\end{tikzpicture}
\caption{Number of models witnessing non-equivalence of formulas found for each universe size by Vampire (blue) and number of models found for each size by the random model generator (red).
\\ Note: The bars are shifted slightly to the sides for visual clarity. }
\end{figure}
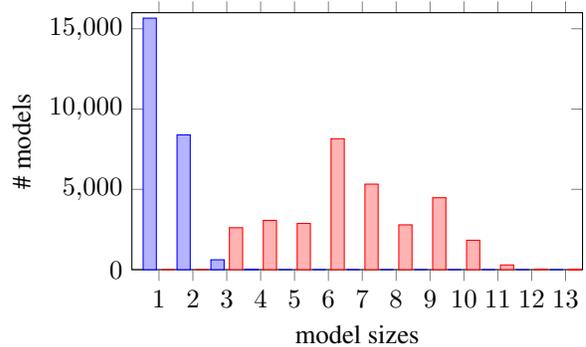

\begin{figure}[t!]
    \begin{tikzpicture}[scale=0.85]
        \begin{axis}[
            name=hist,
            width=\linewidth,
            height=5cm,
            xlabel={time in milliseconds},
            ylabel={\# decisions},
            xticklabel={\pgfmathparse{\tick*10}\pgfmathprintnumber{\pgfmathresult}},
            xmin=0,
            xmax=20,
            ymax=34000,
            enlargelimits=false,
            ytick pos=left,
            scaled y ticks=false,
            yticklabel style={
                /pgf/number format/fixed,
                /pgf/number format/1000 sep={,\!}
            },
        ]

        \addplot[
            ybar,
            draw=blue,
            fill=blue,
            bar width=5pt,
            fill opacity=.2,
            bar shift=-3pt
        ] table [
            x index=0,
            y index=1 ]
        {plot-data/fmbTimeBuckets.txt};

        \addplot[
            ybar,
            draw=red,
            fill=red,
            bar width=5pt,
            fill opacity=.3,
            bar shift=3pt
        ] table [
            x index=0,
            y index=1 ]
        {plot-data/rndTimeBuckets.txt};

        \end{axis}

        \begin{axis}[
            at={(hist.south west)},
            anchor=south west,
            width=\linewidth,
            height=5cm,
xmax=20,
            xmin=0,
            axis x line=none,
            axis y line*=right,
            ylabel={Percentile (\%)},
            ymin=0,
            ymax=100,
ytick pos=right,
]

        \addplot[
            red,
            thick,
            dashed
        ] table [
            x index=0,
            y index=1,
            mark=none,
            each nth point=10
        ] {plot-data/rndPercentiles10.txt};
        
        \addplot[
            blue,
            thick,
            dashed
        ] table [
            x index=0,
            y index=1,
            mark=none,
            each nth point=100
        ] {plot-data/fmbPercentiles10.txt};
        \end{axis}
    \end{tikzpicture}
    \caption{Time needed by Vampire's finite model builder (blue) and by the random model generator (red). The blue and red lines show the percentiles of counter examples found with respect to all distinct non-equivalent pairs of formulas. 
    \\ Note: The bars are shifted slightly to the sides for visual clarity. }
    \label{figure:countermodels:time}
\end{figure}
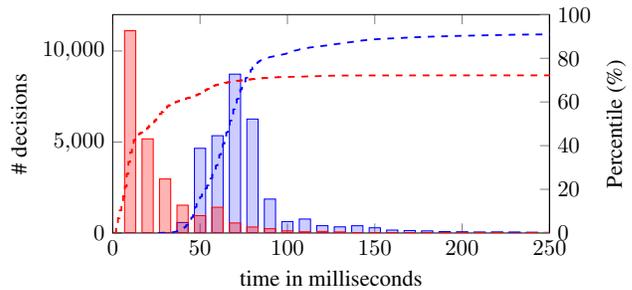

\subsubsection{Results.} Counter examples were found for 99.14\% of the pairs of non-equivalent formulas by at least one of the methods, with Vampire finding counter examples for 97.16\% of all pairs and the random model generator finding counter examples for 80.04\% of all pairs (see \cref{table:summary-all-exercises}). Thus, Vampire could determine non-equivalence of all non-equivalent pairs, but not derive counter examples for a fraction of formula pairs. There are 19.1\% and 1.98\% pairs for which solely Vampire or the random model generator, respectively, found a counter example.

All counter examples found by Vampire were returned within 106ms, also see Figure \ref{figure:countermodels:time}.

In summary, counter examples can be found sufficiently efficiently for most pairs of formulas to provide a first explanation of incorrect student attempts.

\subsection{RQ3: Computing High-level Explanations}
\subsubsection{Setting.}
For answering whether explanations for non-equivalence can be computed in practice, we evaluate (1) which ratio of pairs can be explained by our strategies, and (2) how fast. The expectation, guided by our experience in grading, is that while a significant fraction will be explainable by our strategies, there is also a significant fraction of nonsensical solution attempts.

We implemented all strategies as well as a testing framework to apply them to the data set. Each strategy was applied to every formula pair. In addition to our strategies from Section \ref{section:methods} we also used a strategy Q-1+G-1 that combines modifying differing quantifier prefixes and differing guards. This strategy was added after inspecting pairs that were unexplained by the other strategies.

For determining efficiency in practical scenarios, we ran our experiment also with a cache to store equivalence decisions for pairs of first-order formulas. The cache was initially filled with the equivalence decision of all the 37,159 pairs of formulas. After each Vampire call, the cache was extended with the resulting decisions.

Individual Vampire calls had a timeout of 30s, but the strategies themselves did not have an explicit timeout.

\begin{figure}[t!]
    \begin{tikzpicture}[scale=0.85]
        \begin{axis}[
            name=hist,
            width=\linewidth,
            height=5cm,
            xlabel={time in seconds},
            ylabel={\# decisions},
            xmin=0,
            xmax=100,
            ymax=35000,
            ymode=log,
            enlargelimits=false,
            ytick pos=left,
            scaled y ticks=false,
            yticklabel style={
                /pgf/number format/fixed,
                /pgf/number format/1000 sep={,\!}
            },
        ]

        \hideablePlot{
            \addplot[
                ybar,
                draw=blue,
                fill=blue,
                bar width=1pt,
                fill opacity=.2,
            ] table [
                x index=0,
                y index=1 ]
            {plot-data/stratBuckets-s.txt};
        }

        \hideablePlot{
            \addplot[
                ybar,
                draw=red,
                fill=red,
                bar width=1pt,
                fill opacity=.3,
            ] table [
                x index=0,
                y index=1 ]
            {plot-data/stratCacheBuckets-s.txt};
        }

        \end{axis}

        \begin{axis}[
            at={(hist.south west)},
            anchor=south west,
            width=\linewidth,
            height=5cm,
            xmax=100,
            xmin=0,
            axis x line=none,
            axis y line*=right,
            ylabel={Percentile (\%)},
            ymin=0,
            ymax=100,
            ytick pos=right,
        ]
        
        \hideablePlot{
            \addplot[
                blue,
                thick,
                dashed
            ] table [
                x index=0,
                y index=1,
                mark=none,
                each nth point=10
            ] {plot-data/stratBucketPercentiles-s.txt};
        }

        \hideablePlot{
            \addplot[
                red,
                thick,
                dashed
            ] table [
                x index=0,
                y index=1,
                mark=none,
                each nth point=10
            ] {plot-data/stratCacheBucketPercentiles-s.txt};
        }
        \end{axis}
    \end{tikzpicture}
    \caption{Time in seconds needed to run \emph{all} strategies (sequentially) for each distinct pair of non-equivalent formulas without a cache (blue) and with a cache (red). The blue and red lines show the percentiles for decisions made. Note that the \# decisions axis is in logarithmic scale. Individual Vampire calls had a timeout of 30s, but the strategies themselves did not have an explicit timeout.}
    \label{figure:time:strategies}
\end{figure}
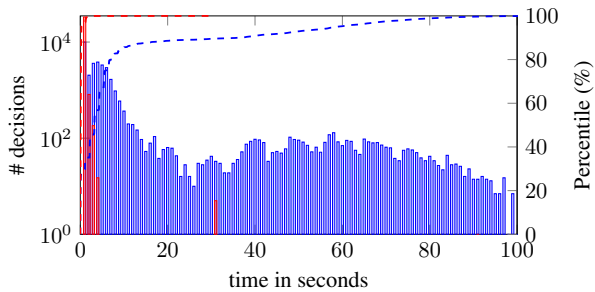
\subsubsection{Results.} Our strategies explain 52.15\% of all non-equivalent formula pairs and 45.72\% of distinct non-equivalent formula pairs.
A detailed split for how many (total and distinct) pairs are explained by the individual strategies is provided in Table \ref{table:summary-all-exercises}.

For 87\% of all pairs, running all strategies sequentially took less than 10 seconds (see Figure \ref{figure:time:strategies}, the \# decisions axis is in logarithmic scale). Using a cache significantly reduces the average runtime from 9.1 seconds without using a cache, to 0.15 seconds, see Figure \ref{figure:time:strategies}. With cache, all but 200 pairs need at most 2 seconds. This indicates that the time needed to run the strategies is dominated by the time needed to run Vampire. Thus, in addition to exploit that pairs of formulas might occur multiple times, a cache is beneficial in practice as also (different) strategies may test the same pairs of formulas multiple times.

In summary, for more than half of all incorrect student attempts, a conceptual explanation can be found. These explanations can also be computed quickly enough to provide students with interactive feedback in an educational support system.

We expect that by fine-tuning the current implementation, the ratio of pairs for which conceptual explanations can be provided can be further increased (e.g.~by adding further strategies and also combining strategies). Also the time needed to compute explanations can likely be decreased by engineering the used algorithms.

\section{Summary, Impact, and Future Work}\label{section:summary}
We designed algorithms for explaining non-equivalence of first-order formulas and evaluated equivalence testing, finding counter models, and the methods for explaining on large educational data sets with $> 100{,}000$ pairs of first-order formulas.

Our evaluation shows that for educational tasks on first-order modelling, (1) modern theorem provers like Vampire can easily determine correctness of modelling attempts, also in the presence of background theories; (2) counter example models can easily be computed for most incorrect modelling attempts by combining Vampire's finite model builder and a random model generator; and (3) conceptual explanations for incorrectness can be computed for a large fraction of modelling attempts ($> 50\%$ of incorrect attempts in our data set). Overall, the computations for (1) – (3) are fast enough to support interactive feedback and to scale to large cohorts.

Our computational strategies for finding explanations are currently being integrated into a state-of-the-art educational support system for learning logic used by $> 1{,}000$ students per year.  We plan to analyse the educational data from this integration using our algorithms to better understand how students learn foundations of logics. 
\section*{Acknowledgements}
We thank Jean Christoph Jung for insightful discussions about Beth definability and interpolation in first-order logic.
This work was supported by the Deutsche Forschungsgemeinschaft (DFG, German Research Foundation), grant 448468041.

\section*{AI Declaration}
Generative AI has only been used for editorial improvements in the writing of this paper.

\bibliographystyle{kr}
\bibliography{bibliography}

\ifwithappendix
\clearpage

\appendix
\section{Details on the Exercises}\label{section:appendix-exercises}
The data used for the evaluation in Section \ref{section:evaluation} originates from exercises for modelling statements in first-order logic. We provide (i) one example of an exercise in Section \ref{section:milisoft}, and (ii) the sample first-order formulas for solutions as well as background theories for all other exercises in Section \ref{section:exercises}.

The data was collected for exercises posed in German, we provide English translations. We only provide the precise statement for one exercise, as the exercises are part of an exercise pool actively used in large logic courses.

\subsection{Sample Exercise: Employees at Millisoft (Exercise E-10)}\label{section:milisoft}

\subsubsection{Assignment text}

    At the web company Millisoft, all employees are either computer scientists or mathematicians.
    The employees work together in teams, with one of the employees acting as team leader. Each employee works in exactly one team.
    The following basic rules apply to the composition of the teams:
    \begin{enumerate}
        \item There is (at least) one mathematician working at Millisoft. 
        \item There is a team leader who is a computer scientist.
        \item Everyone works with at least two computer scientists (one of whom can be themselves) in the same team.
        \item There is one employee whose team consists only of computer scientists.
        \item If the team leader is a mathematician, this does not apply to any other member of the team.
    \end{enumerate}

    In order to check these rules automatically, the aim is to express the rules formally as first-order logic formulas with equality. We model the relevant aspects of the company as a structure \(\mathcal{A}=(A, I^\mathcal{A}, M^\mathcal{A}, G^\mathcal{A}, f^\mathcal{A})\) with the following meaning:

    \begin{itemize}
        \item \(A\) is the universe and consists of all employees of the company
        \item \(I^\mathcal{A} = \{x \in A \mid x \text{ is computer scientist}\}\) is a unary relation
        \item \(M^\mathcal{A} = \{x \in A \mid x \text{ is mathematician}\}\) is a unary relation
        \item \(G^\mathcal{A} = \{(x, y) \in A^2 \mid x \text{ works in the same team as } y\}\) is a binary relation
        \item \(f^\mathcal{A}\) is a unary function that assigns to each employee \(x\) the team leader \(f(x)\) of their team
    \end{itemize}Model the above statements by first-order logic formulas with equality over the signature $\{I, M, G, f\}$ with unary relation symbols $I$ and $M$, binary relation symbol $G$, and unary function symbol $f$ with the intended meaning described above.

\subsubsection{Sample solution}
    \begin{enumerate}
        \item There is (at least) one mathematician working at Millisoft.\\
					\(\exists x\ M(x)\)
        \item There is a team leader who is a computer scientist.\\
        \(\exists x\ \exists y\ ((x=f(y))\wedge I(x))\)
        \item Everyone works with at least two computer scientists (one of whom can be themselves) in the same team.\\
        \(\forall x\ \exists y\ \exists z\ (\neg (y=z)\wedge G(x,y)\wedge G(x,z)\wedge I(y)\wedge I(z))\)
        \item There is one employee whose team consists only of computer scientists.\\
        \(\exists x\ \forall y\ (G(x,y)\rightarrow I(y))\)
        \item If the team leader is a mathematician, this does not apply to any other member of the team.\\
        \(\forall x\ (M(f(x))\rightarrow (\neg (x=f(x))\rightarrow \neg M(x)))\)
    \end{enumerate}

\subsubsection{Background theory}

The background theory models common assumptions and is provided by instructors. It helps students to focus on the statements to be modelled by them.

Background theory for the Millisoft scenario:
\begin{itemize}
 \item Persons in the same team, have the same team leader.\\
 \(\forall x\ \forall y\ (G(x,y)\leftrightarrow (f(x)=f(y)))\)
 \item The domain is a disjoint union of mathematicians and computer scientists.\\
	\(\forall x\ (M(x)\leftrightarrow \neg I(x))\)
 \item Team membership is an equivalence relation.\\
 \(\forall x\ \forall y\ \forall z\ ((G(x,y)\wedge G(y,z))\rightarrow G(x,z))\)

 \(\forall x\ \forall y\ (G(x,y)\rightarrow G(y,x))\)

 \(\forall x\ G(x,x)\)
    \item The team leader of a team member is in the same team.\\
	\(\forall x\ \forall y\ ((f(x)=y)\rightarrow G(x,y))\) \\
\end{itemize}

\subsection{Overview of Exercises: Solution Formulas and Constraints}\label{section:exercises}

Table \ref{table:overview-scenarios-formulas} gives an overview over all scenarios and their sample solutions. Table \ref{table:overview-scenarios-constraints} provides the background theories for all scenarios. A summary of our data set for all exercises is provided in Tables \ref{table:exercise-summary-1} and \ref{table:exercise-summary-2}.

\begin{table*}
    \scriptsize
    \begin{tabular}{lll|l}
	Id  & Scenario & Solution id & Solution formula \\
    \midrule
	E-1 & OnlineShop & \makecell[l]{
E-1-1 \\ 
E-1-2 \\ 
E-1-3} & \makecell[l]{
$\forall x\ (B(x,c)\rightarrow \exists y\ (B(x,y)\wedge \neg (y=c)))$ \\ 
$\exists x\ (V(x,c)\wedge \exists y\ (B(y,x)\wedge \forall z\ (B(z,x)\rightarrow (z=y))))$ \\ 
$\forall x\ \forall y\ \forall z\ ((B(x,y)\wedge B(x,z))\rightarrow (\neg V(y,z)\vee (y=c)\vee (z=c)))$} \\
 \hline
	E-2 & Murder Mystery & \makecell[l]{
E-2-1 \\ 
E-2-2 \\ 
E-2-3 \\ 
E-2-4 \\ 
E-2-5 \\ 
E-2-6 \\ 
E-2-7 \\ 
E-2-8 \\ 
E-2-9 \\ 
E-2-10 \\ 
E-2-11} & \makecell[l]{
$D(a)\wedge D(b)\wedge D(c)$ \\ 
$\forall x\ (D(x)\rightarrow ((x=a)\vee (x=b)\vee (x=c)))$ \\ 
$\exists x\ (D(x)\wedge K(x,a))$ \\ 
$\forall x\ \forall y\ (K(x,y)\rightarrow H(x,y))$ \\ 
$\forall x\ \forall y\ (K(x,y)\rightarrow \neg R(x,y))$ \\ 
$\forall x\ (H(a,x)\rightarrow \neg H(c,x))$ \\ 
$\forall x\ (\neg (x=b)\rightarrow H(a,x))$ \\ 
$\forall x\ (\neg R(x,a)\rightarrow H(b,x))$ \\ 
$\forall x\ (H(a,x)\rightarrow H(b,x))$ \\ 
$\neg \exists x\ \forall y\ H(x,y)$ \\ 
$\neg (a=b)$} \\
 \hline
	E-3 & Film Database & \makecell[l]{
E-3-1 \\ 
E-3-2 \\ 
E-3-3 \\ 
E-3-4} & \makecell[l]{
$\exists x\ L(x,p)$ \\ 
$\forall x\ (F(x)\rightarrow \exists y\ (F(y)\wedge S(x,y)))$ \\ 
$U(x)\wedge \exists y\ \exists z\ (\neg (y=z)\wedge F(y)\wedge F(z)\wedge L(x,y)\wedge L(x,z)\wedge \neg S(y,z))$ \\ 
$F(x)\wedge S(x,p)\wedge \forall y\ ((U(y)\wedge L(y,x))\rightarrow L(y,p))$} \\
 \hline
	E-4 & Social Network & \makecell[l]{
E-4-1 \\ 
E-4-2 \\ 
E-4-3 \\ 
E-4-4} & \makecell[l]{
$\neg I(t)$ \\ 
$\forall x\ (F(x,t)\rightarrow F(t,x))$ \\ 
$\exists x\ (B(x,t)\wedge \forall y\ (B(x,y)\rightarrow (y=t)))$ \\ 
$\forall x\ ((\neg (x=t)\wedge \neg F(x,t))\rightarrow \exists y\ (F(x,y)\wedge F(y,t)))$} \\
 \hline
	E-5 & Volleyball & \makecell[l]{
E-5-1 \\ 
E-5-2 \\ 
E-5-3 \\ 
E-5-4 \\ 
E-5-5} & \makecell[l]{
$T(m)$ \\ 
$\forall x\ \forall y\ \exists z\ ((S(x)\wedge M(y))\rightarrow G(x,y,z))$ \\ 
$\forall x\ \forall y\ (F(x,y)\rightarrow (\neg T(x)\vee \neg T(y)))$ \\ 
$\exists x\ \forall y\ \forall z\ (S(x)\wedge (G(x,y,z)\rightarrow (\neg F(y,z)\wedge \neg (a(y)=z))))$ \\ 
$\exists x\ (M(x)\wedge (a(x)=m))\wedge T(a(m))$} \\
 \hline
	E-6 & Graph & \makecell[l]{
E-6-1 \\ 
E-6-2 \\ 
E-6-3 \\ 
E-6-4 \\ 
E-6-5} & \makecell[l]{
$\forall v\ \neg E(v,v)$ \\ 
$\forall v\ \exists u\ (\neg (v=u)\wedge E(v,u))$ \\ 
$\forall v\ \forall u\ (\neg (u=v)\rightarrow E(v,u))$ \\ 
$\exists v\ \exists u\ \exists w\ (\neg (v=u)\wedge \neg (u=w)\wedge \neg (w=v)\wedge E(v,u)\wedge E(u,w)\wedge E(w,v))$ \\ 
$\forall v\ (\exists u\ E(u,v)\rightarrow \exists u\ E(v,u))$} \\
 \hline
	E-7 & Murder Mystery II & \makecell[l]{
E-7-1 \\ 
E-7-2 \\ 
E-7-3} & \makecell[l]{
$\forall x\ (D(x)\rightarrow \exists y\ M(y,x))$ \\ 
$\forall x\ (H(x,x)\rightarrow \neg \exists y\ M(x,y))$ \\ 
$\exists x\ (D(x)\wedge R(x,c))$} \\
 \hline
	E-8 & Strings & \makecell[l]{
E-8-1 \\ 
E-8-2 \\ 
E-8-3} & \makecell[l]{
$\forall x\ \forall y\ (\neg Z_c(x)\wedge ((Z_a(x)\wedge Z_b(y))\rightarrow (x\le y)))$ \\ 
$\exists x\ ((l(x)=x)\wedge \neg (x=r(x))\wedge \neg (r(x)=r(r(x)))\wedge Z_a(x)\wedge Z_b(r(x))\wedge Z_c(r(r(x))))$ \\ 
$\forall x\ \forall y\ (((r(x)=y)\wedge \neg (x=y))\rightarrow ((Z_a(x)\wedge Z_b(y))\vee (Z_b(x)\wedge Z_a(y))))$} \\
 \hline
	E-9 & Theorems & \makecell[l]{
E-9-1 \\ 
E-9-2 \\ 
E-9-3 \\ 
E-9-4} & \makecell[l]{
$\forall x\ \forall y\ \forall z\ (((x\circ y)\circ z)=(x\circ (y\circ z)))$ \\ 
$\forall x\ ((e\circ x)=x)$ \\ 
$\forall x\ \exists y\ ((y\circ x)=e)$ \\ 
$\forall x\ \forall y\ \exists z\ ((x\circ z)=y)$} \\
 \hline
	E-10 & Millisoft & \makecell[l]{
E-10-1 \\ 
E-10-2 \\ 
E-10-3 \\ 
E-10-4 \\ 
E-10-5} & \makecell[l]{
$\exists x\ M(x)$ \\ 
$\exists x\ \exists y\ ((x=f(y))\wedge I(x))$ \\ 
$\forall x\ \exists y\ \exists z\ (\neg (y=z)\wedge G(x,y)\wedge G(x,z)\wedge I(y)\wedge I(z))$ \\ 
$\exists x\ \forall y\ (G(x,y)\rightarrow I(y))$ \\ 
$\forall x\ (M(f(x))\rightarrow (\neg (x=f(x))\rightarrow \neg M(x)))$} \\
 \hline
	E-11 & Web pages & \makecell[l]{
E-11-1 \\ 
E-11-2 \\ 
E-11-3 \\ 
E-11-4} & \makecell[l]{
$\forall x\ (M(x)\rightarrow \neg V(s(x)))$ \\ 
$\forall x\ (E(x)\rightarrow (\exists y\ (M(y)\wedge L(x,y))\vee \exists y\ (N(y)\wedge L(y,x))))$ \\ 
$\exists y\ (L(s(p),y)\wedge L(y,e))$ \\ 
$\neg \exists x\ (N(x)\wedge \neg \exists y\ L(x,y))$} \\
 \hline
	E-12 & Recipe & \makecell[l]{
E-12-1 \\ 
E-12-2 \\ 
E-12-3 \\ 
E-12-4} & \makecell[l]{
$\exists x\ R(x)\wedge \exists x\ Z(x)$ \\ 
$\exists x\ (R(x)\wedge (f(x)=v)\wedge \exists y\ E(x,y))$ \\ 
$\exists x\ (R(x)\wedge \exists y\ (E(x,y)\wedge \neg \exists z\ (E(x,z)\wedge \neg (z=y))))$ \\ 
$\forall x\ \forall y\ ((R(x)\wedge R(y)\wedge (f(x)=f(y)))\rightarrow \exists z\ (E(x,z)\wedge E(y,z)))$} \\
 \hline
	E-13 & Book Collection & \makecell[l]{
E-13-1 \\ 
E-13-2 \\ 
E-13-3 \\ 
E-13-4} & \makecell[l]{
$M(f(p))\wedge L(p)$ \\ 
$\forall x\ \forall y\ ((W(x,y)\wedge L(y))\rightarrow L(x))$ \\ 
$\exists x\ (B(x)\wedge \neg M(f(x))\wedge \forall y\ \neg W(y,x))$ \\ 
$\forall x\ (A(x)\rightarrow \exists y\ ((f(y)=x)\wedge B(y)\wedge (L(y)\vee \exists z\ W(z,y))))$} \\
 \hline
	E-14 & Faculty Conference & \makecell[l]{
E-14-1 \\ 
E-14-2 \\ 
E-14-3} & \makecell[l]{
$\exists x\ (P(x)\wedge (g(x)=l)\wedge \forall y\ (K(y)\rightarrow \neg B(x,y)))$ \\ 
$\forall x\ \forall y\ ((K(x)\wedge A(y))\rightarrow \neg \exists w\ \exists z\ (\neg (w=z)\wedge (g(w)=y)\wedge (g(z)=y)\wedge B(w,x)\wedge B(z,x)))$ \\ 
$\exists x\ (A(x)\wedge \forall y\ (K(y)\rightarrow (\forall z\ ((P(z)\wedge (g(z)=x))\rightarrow (B(z,x)\vee V(z,x)))\vee \neg \exists z\ (P(z)\wedge (g(z)=x)\wedge (B(z,x)\vee V(z,x))))))$} \\
    \bottomrule
\end{tabular}
\caption{Exercise scenarios with solution formulas. We use the convention that $ x $, $ y $, $ z $ are variables, capital letters are relation symbols, and lower-case letters are function or constant symbols.
}\label{table:overview-scenarios-formulas}
\end{table*}

\begin{table*}
    \scriptsize
    \begin{tabular}{ll|l}
	Id & Scenario & Constraints \\
    \midrule
	E-1 & Online Shop & \makecell[l]{
$\forall x\ ((A(x)\wedge \neg N(x))\vee (\neg A(x)\wedge N(x)))$ \\ 
$A(c)$ \\ 
$\forall x\ \forall y\ (B(x,y)\rightarrow (N(x)\wedge A(y)))$ \\ 
$\forall x\ \forall y\ (V(x,y)\rightarrow (A(x)\wedge A(y)\wedge V(y,x)\wedge \neg (x=y)))$} \\
 \hline
	E-2 & Murder Mystery & \makecell[l]{
$\forall x\ \neg R(x,x)$ \\ 
$\forall x\ \forall y\ (R(x,y)\rightarrow \neg R(y,x))$} \\
 \hline
	E-3 & Film Database & \makecell[l]{
$\forall f\ \forall g\ (S(f,g)\rightarrow S(g,f))$ \\ 
$\forall f\ \neg S(f,f)$ \\ 
$\forall f\ \forall g\ (S(f,g)\rightarrow (F(f)\wedge F(g)))$ \\ 
$\forall u\ \forall f\ (L(u,f)\rightarrow (U(u)\wedge F(f)))$ \\ 
$\forall x\ ((U(x)\wedge \neg F(x))\vee (F(x)\wedge \neg U(x)))$ \\ 
$F(p)$} \\
 \hline
	E-4 & Social Network & \makecell[l]{
$\forall x\ \forall y\ (B(x,y)\rightarrow \neg (F(x,y)\vee F(y,x)))$ \\ 
$\forall x\ \neg F(x,x)$ \\ 
$\forall x\ \neg B(x,x)$} \\
 \hline
	E-5 & Volleyball & \makecell[l]{
$\forall x\ (\neg S(x)\leftrightarrow M(x))$ \\ 
$\forall x\ \forall y\ \forall z\ (G(x,y,z)\rightarrow G(x,z,y))$ \\ 
$\forall x\ \forall y\ \forall z\ (G(x,y,z)\rightarrow \neg (y=z))$ \\ 
$\forall x\ \forall y\ (F(x,y)\rightarrow F(y,x))$ \\ 
$\forall x\ \forall y\ (F(x,y)\rightarrow \neg (x=y))$ \\ 
$\forall x\ (T(x)\rightarrow M(x))$ \\ 
$\forall x\ \forall y\ \forall z\ (G(x,y,z)\rightarrow (S(x)\wedge M(y)\wedge M(z)))$ \\ 
$\forall x\ \forall y\ (F(x,y)\rightarrow (M(x)\wedge M(y)))$ \\ 
$\forall x\ \exists y\ (M(x)\rightarrow ((a(x)=y)\wedge M(y)\wedge \neg (x=y)))$ \\ 
$M(m)$} \\
 \hline
	E-6 & Graph &  \\
 \hline
	E-7 & Murder Mystery II & \makecell[l]{
$\forall x\ \neg R(x,x)$ \\ 
$\forall x\ \forall y\ (R(x,y)\rightarrow \neg R(y,x))$ \\ 
$D(a)\wedge D(b)\wedge D(c)$} \\
 \hline
	E-8 & Strings & \makecell[l]{
$\forall x\ (Z_a(x)\vee Z_b(x)\vee Z_c(x))$ \\ 
$\forall x\ \neg (Z_a(x)\wedge Z_b(x))$ \\ 
$\forall x\ \neg (Z_a(x)\wedge Z_c(x))$ \\ 
$\forall x\ \neg (Z_b(x)\wedge Z_c(x))$ \\ 
$\forall x\ (x\le x)$ \\ 
$\forall x\ \forall y\ (((x\le y)\wedge (y\le x))\rightarrow (x=y))$ \\ 
$\forall x\ \forall y\ \forall z\ (((x\le y)\wedge (y\le z))\rightarrow (x\le z))$ \\ 
$\forall x\ \forall y\ ((x\le y)\vee (y\le x))$ \\ 
$\exists x\ ((x=l(x))\wedge \neg \exists y\ ((l(y)=y)\wedge \neg (x=y))\wedge \forall z\ (x\le z))$ \\ 
$\exists x\ ((x=r(x))\wedge \neg \exists y\ ((r(y)=y)\wedge \neg (x=y))\wedge \forall z\ (z\le x))$ \\ 
$\forall x\ (\neg (x=l(x))\rightarrow ((l(x)\le x)\wedge \neg \exists y\ ((l(x)\le y)\wedge (y\le x)\wedge \neg (l(x)=y)\wedge \neg (y=x))))$ \\ 
$\forall x\ (\neg (x=r(x))\rightarrow ((x\le r(x))\wedge \neg \exists y\ ((x\le y)\wedge (y\le r(x))\wedge \neg (x=y)\wedge \neg (y=r(x)))))$} \\
 \hline
	E-9 & Theorems &  \\
 \hline
	E-10 & Millisoft & \makecell[l]{
$\forall x\ \forall y\ (G(x,y)\leftrightarrow (f(x)=f(y)))$ \\ 
$\forall x\ (M(x)\leftrightarrow \neg I(x))$ \\ 
$\forall x\ \forall y\ \forall z\ ((G(x,y)\wedge G(y,z))\rightarrow G(x,z))$ \\ 
$\forall x\ \forall y\ (G(x,y)\rightarrow G(y,x))$ \\ 
$\forall x\ G(x,x)$ \\ 
$\forall x\ \forall y\ ((f(x)=y)\rightarrow G(x,y))$} \\
 \hline
	E-11 & Web pages &  \\
 \hline
	E-12 & Recipe & \makecell[l]{
$\forall x\ ((R(x)\wedge \neg Z(x)\wedge \neg K(x))\vee (\neg R(x)\wedge Z(x)\wedge \neg K(x))\vee (\neg R(x)\wedge \neg Z(x)\wedge K(x)))$ \\ 
$K(v)$ \\ 
$\forall x\ (R(x)\rightarrow K(f(x)))$ \\ 
$\forall x\ (\neg R(x)\rightarrow (f(x)=x))$ \\ 
$\forall x\ \forall y\ (E(x,y)\rightarrow (R(x)\wedge Z(y)))$} \\
 \hline
	E-13 & Book Collection & \makecell[l]{
$B(p)$ \\ 
$\forall x\ (A(x)\vee B(x))$ \\ 
$\forall x\ (A(x)\leftrightarrow \neg B(x))$ \\ 
$\forall x\ A(f(x))$ \\ 
$\forall x\ (M(x)\rightarrow A(x))$ \\ 
$\forall x\ (L(x)\rightarrow B(x))$ \\ 
$\forall x\ \forall y\ (W(x,y)\rightarrow (B(x)\wedge B(y)))$} \\
 \hline
	E-14 & Faculty Conference & \makecell[l]{
$A(l)$ \\ 
$\forall x\ (A(x)\vee P(x)\vee K(x))$ \\ 
$\forall x\ (\neg (A(x)\wedge P(x))\wedge \neg (A(x)\wedge K(x))\wedge \neg (P(x)\wedge K(x)))$ \\ 
$\forall x\ \forall y\ ((B(x,y)\vee V(x,y))\rightarrow (P(x)\wedge K(y)))$ \\ 
$\forall x\ (P(x)\rightarrow A(g(x)))$} \\
    \bottomrule
\end{tabular}
\caption{
Exercise scenarios with background theories. We use the convention that $ x $, $ y $, $ z $ are variables, capital letters are relation symbols, and lower-case letters are function or constant symbols.}\label{table:overview-scenarios-constraints}
\end{table*}

\renewcommand{\fiiill}{&&&&&&&&&& &&&&}
\begin{table*}
    \footnotesize
    \centering
    \begin{tabular}{l| l || rrrr | rrrr | rrrr}
        \toprule 
        Exercise id & Exercise \& formulas & \multicolumn{4}{c|}{all attempts} & \multicolumn{4}{c|}{equivalent attempts} & \multicolumn{4}{c}{non-equivalent attempts}\\
        && \defaultColumnHeaders & \defaultColumnHeaders & \defaultColumnHeaders \\
        \midrule
        & All exercises combined & 125970 & 100.0\% & 37159 & 100.0\% & 27655 & 21.95\% & 3019 & 8.12\% & 98315 & 78.05\% & 34140 & 91.88\% \\
        \midrule
        	E-1 & OnlineShop & 6746 & 5.36\% & 4235 & 11.4\% & 1222 & 18.11\% & 414 & 9.78\% & 5524 & 81.89\% & 3821 & 90.22\% \\
	& \tableindent{} E-1-1 & 2327 & 34.49\% & 1185 & 27.98\% & 517 & 22.22\% & 124 & 10.46\% & 1810 & 77.78\% & 1061 & 89.54\% \\
	& \tableindent{} E-1-2 & 1619 & 24.0\% & 1038 & 24.51\% & 468 & 28.91\% & 166 & 15.99\% & 1151 & 71.09\% & 872 & 84.01\% \\
	& \tableindent{} E-1-3 & 2800 & 41.51\% & 2012 & 47.51\% & 237 & 8.46\% & 124 & 6.16\% & 2563 & 91.54\% & 1888 & 93.84\% \\
	E-2 & Murder Mystery & 4525 & 3.59\% & 782 & 2.1\% & 2353 & 52.0\% & 83 & 10.61\% & 2172 & 48.0\% & 699 & 89.39\% \\
	& \tableindent{} E-2-1 & 348 & 7.69\% & 51 & 6.52\% & 264 & 75.86\% & 12 & 23.53\% & 84 & 24.14\% & 39 & 76.47\% \\
	& \tableindent{} E-2-2 & 663 & 14.65\% & 160 & 20.46\% & 201 & 30.32\% & 8 & 5.0\% & 462 & 69.68\% & 152 & 95.0\% \\
	& \tableindent{} E-2-3 & 496 & 10.96\% & 76 & 9.72\% & 229 & 46.17\% & 5 & 6.58\% & 267 & 53.83\% & 71 & 93.42\% \\
	& \tableindent{} E-2-4 & 524 & 11.58\% & 72 & 9.21\% & 218 & 41.6\% & 2 & 2.78\% & 306 & 58.4\% & 70 & 97.22\% \\
	& \tableindent{} E-2-5 & 382 & 8.44\% & 56 & 7.16\% & 211 & 55.24\% & 7 & 12.5\% & 171 & 44.76\% & 49 & 87.5\% \\
	& \tableindent{} E-2-6 & 385 & 8.51\% & 87 & 11.13\% & 212 & 55.06\% & 11 & 12.64\% & 173 & 44.94\% & 76 & 87.36\% \\
	& \tableindent{} E-2-7 & 475 & 10.5\% & 92 & 11.76\% & 204 & 42.95\% & 11 & 11.96\% & 271 & 57.05\% & 81 & 88.04\% \\
	& \tableindent{} E-2-8 & 318 & 7.03\% & 46 & 5.88\% & 205 & 64.47\% & 5 & 10.87\% & 113 & 35.53\% & 41 & 89.13\% \\
	& \tableindent{} E-2-9 & 270 & 5.97\% & 29 & 3.71\% & 207 & 76.67\% & 5 & 17.24\% & 63 & 23.33\% & 24 & 82.76\% \\
	& \tableindent{} E-2-10 & 402 & 8.88\% & 84 & 10.74\% & 194 & 48.26\% & 9 & 10.71\% & 208 & 51.74\% & 75 & 89.29\% \\
	& \tableindent{} E-2-11 & 262 & 5.79\% & 29 & 3.71\% & 208 & 79.39\% & 8 & 27.59\% & 54 & 20.61\% & 21 & 72.41\% \\
	E-3 & FilmDatabase & 5938 & 4.71\% & 2227 & 5.99\% & 1435 & 24.17\% & 275 & 12.35\% & 4503 & 75.83\% & 1952 & 87.65\% \\
	& \tableindent{} E-3-1 & 552 & 9.3\% & 100 & 4.49\% & 413 & 74.82\% & 24 & 24.0\% & 139 & 25.18\% & 76 & 76.0\% \\
	& \tableindent{} E-3-2 & 1645 & 27.7\% & 283 & 12.71\% & 386 & 23.47\% & 34 & 12.01\% & 1259 & 76.53\% & 249 & 87.99\% \\
	& \tableindent{} E-3-3 & 2429 & 40.91\% & 1180 & 52.99\% & 334 & 13.75\% & 142 & 12.03\% & 2095 & 86.25\% & 1038 & 87.97\% \\
	& \tableindent{} E-3-4 & 1312 & 22.09\% & 664 & 29.82\% & 302 & 23.02\% & 75 & 11.3\% & 1010 & 76.98\% & 589 & 88.7\% \\
	E-4 & SocialNetwork & 6119 & 4.86\% & 1711 & 4.6\% & 1615 & 26.39\% & 133 & 7.77\% & 4504 & 73.61\% & 1578 & 92.23\% \\
	& \tableindent{} E-4-1 & 717 & 11.72\% & 47 & 2.75\% & 464 & 64.71\% & 9 & 19.15\% & 253 & 35.29\% & 38 & 80.85\% \\
	& \tableindent{} E-4-2 & 864 & 14.12\% & 173 & 10.11\% & 436 & 50.46\% & 8 & 4.62\% & 428 & 49.54\% & 165 & 95.38\% \\
	& \tableindent{} E-4-3 & 3545 & 57.93\% & 1075 & 62.83\% & 359 & 10.13\% & 64 & 5.95\% & 3186 & 89.87\% & 1011 & 94.05\% \\
	& \tableindent{} E-4-4 & 993 & 16.23\% & 416 & 24.31\% & 356 & 35.85\% & 52 & 12.5\% & 637 & 64.15\% & 364 & 87.5\% \\
	E-5 & Volleyball & 6370 & 5.06\% & 3005 & 8.09\% & 1424 & 22.35\% & 393 & 13.08\% & 4946 & 77.65\% & 2612 & 86.92\% \\
	& \tableindent{} E-5-1 & 417 & 6.55\% & 31 & 1.03\% & 391 & 93.76\% & 14 & 45.16\% & 26 & 6.24\% & 17 & 54.84\% \\
	& \tableindent{} E-5-2 & 2686 & 42.17\% & 973 & 32.38\% & 290 & 10.8\% & 71 & 7.3\% & 2396 & 89.2\% & 902 & 92.7\% \\
	& \tableindent{} E-5-3 & 978 & 15.35\% & 506 & 16.84\% & 311 & 31.8\% & 76 & 15.02\% & 667 & 68.2\% & 430 & 84.98\% \\
	& \tableindent{} E-5-4 & 1541 & 24.19\% & 1175 & 39.1\% & 211 & 13.69\% & 160 & 13.62\% & 1330 & 86.31\% & 1015 & 86.38\% \\
	& \tableindent{} E-5-5 & 748 & 11.74\% & 320 & 10.65\% & 221 & 29.55\% & 72 & 22.5\% & 527 & 70.45\% & 248 & 77.5\% \\
	E-6 & Graph & 40362 & 32.04\% & 5844 & 15.73\% & 7235 & 17.93\% & 265 & 4.53\% & 33127 & 82.07\% & 5579 & 95.47\% \\
	& \tableindent{} E-6-1 & 4618 & 11.44\% & 513 & 8.78\% & 1836 & 39.76\% & 52 & 10.14\% & 2782 & 60.24\% & 461 & 89.86\% \\
	& \tableindent{} E-6-2 & 9048 & 22.42\% & 1096 & 18.75\% & 1545 & 17.08\% & 26 & 2.37\% & 7503 & 82.92\% & 1070 & 97.63\% \\
	& \tableindent{} E-6-3 & 9571 & 23.71\% & 1070 & 18.31\% & 1370 & 14.31\% & 25 & 2.34\% & 8201 & 85.69\% & 1045 & 97.66\% \\
	& \tableindent{} E-6-4 & 4760 & 11.79\% & 971 & 16.62\% & 1281 & 26.91\% & 105 & 10.81\% & 3479 & 73.09\% & 866 & 89.19\% \\
	& \tableindent{} E-6-5 & 12365 & 30.64\% & 2194 & 37.54\% & 1203 & 9.73\% & 57 & 2.6\% & 11162 & 90.27\% & 2137 & 97.4\% \\
	E-7 & Murder Mystery-II & 987 & 0.78\% & 215 & 0.58\% & 294 & 29.79\% & 16 & 7.44\% & 693 & 70.21\% & 199 & 92.56\% \\
	& \tableindent{} E-7-1 & 497 & 50.35\% & 105 & 48.84\% & 96 & 19.32\% & 5 & 4.76\% & 401 & 80.68\% & 100 & 95.24\% \\
	& \tableindent{} E-7-2 & 268 & 27.15\% & 78 & 36.28\% & 93 & 34.7\% & 5 & 6.41\% & 175 & 65.3\% & 73 & 93.59\% \\
	& \tableindent{} E-7-3 & 222 & 22.49\% & 32 & 14.88\% & 105 & 47.3\% & 6 & 18.75\% & 117 & 52.7\% & 26 & 81.25\% \\
    \bottomrule
    \end{tabular}
    \caption{Summary of collected data for exercises E-1 -- E-7. Percentages in the all-attempts columns for individual formulas are relative to their exercise's totals; percentages in the equivalent and non-equivalent columns are relative to all attempts for the corresponding formula (analogously for distinct attempts).}\label{table:exercise-summary-1}
\end{table*}

\begin{table*}
    \footnotesize
    \centering
    \begin{tabular}{l | l || rrrr | rrrr | rrrr}
        \toprule 
        Id & Exercise \& formulas & \multicolumn{4}{c|}{all attempts} & \multicolumn{4}{c|}{equivalent attempts} & \multicolumn{4}{c}{non-equivalent attempts}\\
        && \defaultColumnHeaders & \defaultColumnHeaders & \defaultColumnHeaders \\
        \midrule
         & All exercises combined & 125970 & 100.0\% & 37159 & 100.0\% & 27655 & 21.95\% & 3019 & 8.12\% & 98315 & 78.05\% & 34140 & 91.88\% \\
        \midrule
        	E-8 & Strings & 640 & 0.51\% & 289 & 0.78\% & 87 & 13.59\% & 41 & 14.19\% & 553 & 86.41\% & 248 & 85.81\% \\
	& \tableindent{} E-8-1 & 160 & 25.0\% & 74 & 25.61\% & 41 & 25.63\% & 15 & 20.27\% & 119 & 74.38\% & 59 & 79.73\% \\
	& \tableindent{} E-8-2 & 175 & 27.34\% & 102 & 35.29\% & 32 & 18.29\% & 18 & 17.65\% & 143 & 81.71\% & 84 & 82.35\% \\
	& \tableindent{} E-8-3 & 305 & 47.66\% & 113 & 39.1\% & 14 & 4.59\% & 8 & 7.08\% & 291 & 95.41\% & 105 & 92.92\% \\
	E-9 & Theorems & 488 & 0.39\% & 114 & 0.31\% & 176 & 36.07\% & 15 & 13.16\% & 312 & 63.93\% & 99 & 86.84\% \\
	& \tableindent{} E-9-1 & 43 & 8.81\% & 5 & 4.39\% & 39 & 90.7\% & 2 & 40.0\% & 4 & 9.3\% & 3 & 60.0\% \\
	& \tableindent{} E-9-2 & 146 & 29.92\% & 38 & 33.33\% & 50 & 34.25\% & 7 & 18.42\% & 96 & 65.75\% & 31 & 81.58\% \\
	& \tableindent{} E-9-3 & 193 & 39.55\% & 40 & 35.09\% & 47 & 24.35\% & 3 & 7.5\% & 146 & 75.65\% & 37 & 92.5\% \\
	& \tableindent{} E-9-4 & 106 & 21.72\% & 31 & 27.19\% & 40 & 37.74\% & 3 & 9.68\% & 66 & 62.26\% & 28 & 90.32\% \\
	E-10 & Millisoft & 28095 & 22.3\% & 9336 & 25.12\% & 5220 & 18.58\% & 774 & 8.29\% & 22875 & 81.42\% & 8562 & 91.71\% \\
	& \tableindent{} E-10-1 & 1888 & 6.72\% & 126 & 1.35\% & 1319 & 69.86\% & 19 & 15.08\% & 569 & 30.14\% & 107 & 84.92\% \\
	& \tableindent{} E-10-2 & 2694 & 9.59\% & 483 & 5.17\% & 1206 & 44.77\% & 61 & 12.63\% & 1488 & 55.23\% & 422 & 87.37\% \\
	& \tableindent{} E-10-3 & 7170 & 25.52\% & 2949 & 31.59\% & 933 & 13.01\% & 311 & 10.55\% & 6237 & 86.99\% & 2638 & 89.45\% \\
	& \tableindent{} E-10-4 & 3906 & 13.9\% & 981 & 10.51\% & 990 & 25.35\% & 53 & 5.4\% & 2916 & 74.65\% & 928 & 94.6\% \\
	& \tableindent{} E-10-5 & 12437 & 44.27\% & 4797 & 51.38\% & 772 & 6.21\% & 330 & 6.88\% & 11665 & 93.79\% & 4467 & 93.12\% \\
	E-11 & Webpages & 11543 & 9.16\% & 3771 & 10.15\% & 2784 & 24.12\% & 170 & 4.51\% & 8759 & 75.88\% & 3601 & 95.49\% \\
	& \tableindent{} E-11-1 & 4604 & 39.89\% & 1261 & 33.44\% & 752 & 16.33\% & 40 & 3.17\% & 3852 & 83.67\% & 1221 & 96.83\% \\
	& \tableindent{} E-11-2 & 2148 & 18.61\% & 845 & 22.41\% & 691 & 32.17\% & 43 & 5.09\% & 1457 & 67.83\% & 802 & 94.91\% \\
	& \tableindent{} E-11-3 & 2945 & 25.51\% & 1198 & 31.77\% & 659 & 22.38\% & 65 & 5.43\% & 2286 & 77.62\% & 1133 & 94.57\% \\
	& \tableindent{} E-11-4 & 1846 & 15.99\% & 467 & 12.38\% & 682 & 36.94\% & 22 & 4.71\% & 1164 & 63.06\% & 445 & 95.29\% \\
	E-12 & Recipe & 11200 & 8.89\% & 3778 & 10.17\% & 3388 & 30.25\% & 292 & 7.73\% & 7812 & 69.75\% & 3486 & 92.27\% \\
	& \tableindent{} E-12-1 & 1756 & 15.68\% & 165 & 4.37\% & 1024 & 58.31\% & 19 & 11.52\% & 732 & 41.69\% & 146 & 88.48\% \\
	& \tableindent{} E-12-2 & 2431 & 21.71\% & 740 & 19.59\% & 947 & 38.96\% & 71 & 9.59\% & 1484 & 61.04\% & 669 & 90.41\% \\
	& \tableindent{} E-12-3 & 3819 & 34.1\% & 1735 & 45.92\% & 803 & 21.03\% & 150 & 8.65\% & 3016 & 78.97\% & 1585 & 91.35\% \\
	& \tableindent{} E-12-4 & 3194 & 28.52\% & 1138 & 30.12\% & 614 & 19.22\% & 52 & 4.57\% & 2580 & 80.78\% & 1086 & 95.43\% \\
	E-13 & BookCollection & 2040 & 1.62\% & 1230 & 3.31\% & 345 & 16.91\% & 112 & 9.11\% & 1695 & 83.09\% & 1118 & 90.89\% \\
	& \tableindent{} E-13-1 & 267 & 13.09\% & 126 & 10.24\% & 118 & 44.19\% & 24 & 19.05\% & 149 & 55.81\% & 102 & 80.95\% \\
	& \tableindent{} E-13-2 & 582 & 28.53\% & 355 & 28.86\% & 92 & 15.81\% & 27 & 7.61\% & 490 & 84.19\% & 328 & 92.39\% \\
	& \tableindent{} E-13-3 & 582 & 28.53\% & 365 & 29.67\% & 86 & 14.78\% & 40 & 10.96\% & 496 & 85.22\% & 325 & 89.04\% \\
	& \tableindent{} E-13-4 & 609 & 29.85\% & 384 & 31.22\% & 49 & 8.05\% & 21 & 5.47\% & 560 & 91.95\% & 363 & 94.53\% \\
	E-14 & FacultyConference & 917 & 0.73\% & 622 & 1.67\% & 77 & 8.4\% & 36 & 5.79\% & 840 & 91.6\% & 586 & 94.21\% \\
	& \tableindent{} E-14-1 & 607 & 66.19\% & 379 & 60.93\% & 44 & 7.25\% & 15 & 3.96\% & 563 & 92.75\% & 364 & 96.04\% \\
	& \tableindent{} E-14-2 & 235 & 25.63\% & 186 & 29.9\% & 16 & 6.81\% & 12 & 6.45\% & 219 & 93.19\% & 174 & 93.55\% \\
	& \tableindent{} E-14-3 & 75 & 8.18\% & 57 & 9.16\% & 17 & 22.67\% & 9 & 15.79\% & 58 & 77.33\% & 48 & 84.21\% \\
    \bottomrule
    \end{tabular}
    \caption{Summary of collected data for exercises E-8 -- E-14. Percentages in the all-attempts columns for individual formulas are relative to their exercise's totals; percentages in the equivalent and non-equivalent columns are relative to all attempts for the corresponding formula (analogously for distinct attempts).}\label{table:exercise-summary-2}
\end{table*}
 
\clearpage
\onecolumn
\section{Detailed Evaluation Data}\label{section:appendix-evaluation-tables}
The detailed evaluation data for all exercises can be found in Tables \ref{table:exercise-online-shop},
\ref{table:exercise-dreadsbury},
\ref{table:exercise-dreadsbury-pt2},
\ref{table:exercise-film-database},
\ref{table:exercise-social-network},
\ref{table:exercise-volleyball},
\ref{table:exercise-graph-ii},
\ref{table:exercise-dreadsbury-ii},
\ref{table:exercise-strings},
\ref{table:exercise-theorems},
\ref{table:exercise-millisoft},
\ref{table:exercise-web-pages},
\ref{table:exercise-recipe},
\ref{table:exercise-book-collection}, and
\ref{table:exercise-faculty-conference}.

\renewcommand{\fiiill}{&&&&&&&&&& &&&&}
\begin{table*}
    \footnotesize
    \centering
    \begin{adjustbox}{angle=90}
\begin{tabular}{l|| rrrr | rrrr | rrrr | rrrr}
        \toprule 
        Exercise E-1: Online Shop & \multicolumn{4}{c|}{combined} & \multicolumn{4}{c|}{formula 1} & \multicolumn{4}{c|}{formula 2} & \multicolumn{4}{c}{formula 3}\\
        & \defaultColumnHeaders & \defaultColumnHeaders & \defaultColumnHeaders & \defaultColumnHeaders \\
        \midrule
            \midrule
	all attempts & 6746 & 100.0\% & 4235 & 100.0\% & 2327 & 34.49\% & 1185 & 27.98\% & 1619 & 24.0\% & 1038 & 24.51\% & 2800 & 41.51\% & 2012 & 47.51\% \\
	equivalent attempts & 1222 & 18.11\% & 414 & 9.78\% & 517 & 22.22\% & 124 & 10.46\% & 468 & 28.91\% & 166 & 15.99\% & 237 & 8.46\% & 124 & 6.16\% \\
	non-equivalent attempts & 5524 & 81.89\% & 3821 & 90.22\% & 1810 & 77.78\% & 1061 & 89.54\% & 1151 & 71.09\% & 872 & 84.01\% & 2563 & 91.54\% & 1888 & 93.84\% \\
    \midrule
\multicolumn{17}{c}{\textbf{Equivalence}$ ^* $}\\
\midrule
	proven by Vampire (\texttt{vampire}, \texttt{casc}, \texttt{casc\_sat} combined) & 1222 & 100.0\% & 414 & 100.0\% & 517 & 100.0\% & 124 & 100.0\% & 468 & 100.0\% & 166 & 100.0\% & 237 & 100.0\% & 124 & 100.0\% \\
    \midrule
\multicolumn{17}{c}{\textbf{Non-Equivalence}$ ^{**} $}\\
\bottomrule
    \end{tabular}
    \end{adjustbox}
    \caption{Evaluation for exercise E-1 Online Shop.}\label{table:exercise-online-shop}
\end{table*}

\renewcommand{\fiiill}{&&&&&&&&&&&&&&&&&& &&&&}
\begin{table*}
    \scriptsize
    \centering
    \begin{adjustbox}{angle=90}
    \begin{tabular}{l|| rrrr | rrrr | rrrr | rrrr | rrrr | rrrr}
        \toprule 
        Exercise E-2: Murder mystery (E-2-1 - E-2-5) & \multicolumn{4}{c|}{combined} & \multicolumn{4}{c|}{formula 1} & \multicolumn{4}{c|}{formula 2} & \multicolumn{4}{c|}{formula 3} & \multicolumn{4}{c|}{formula 4} & \multicolumn{4}{c}{formula 5}\\
        & \defaultColumnHeaders & \defaultColumnHeaders & \defaultColumnHeaders & \defaultColumnHeaders & \defaultColumnHeaders & \defaultColumnHeaders \\
\midrule
all attempts                                                                      & 4525 & 100.0\% & 782     & 100.0\% & 348  & 7.69\%  & 51      & 6.52\%  & 663  & 14.65\% & 160     & 20.46\% & 496  & 10.96\% & 76      & 9.72\%  & 524  & 11.58\% & 72      & 9.21\%  & 382  & 8.44\%  & 56      & 7.16\% \\
equivalent attempts                                                               & 2353 & 52.0\%  & 83      & 10.61\% & 264  & 75.86\% & 12      & 23.53\% & 201  & 30.32\% & 8       & 5.0\%   & 229  & 46.17\% & 5       & 6.58\%  & 218  & 41.6\%  & 2       & 2.78\%  & 211  & 55.24\% & 7       & 12.5\% \\
non-equivalent attempts                                                           & 2172 & 48.0\%  & 699     & 89.39\% & 84   & 24.14\% & 39      & 76.47\% & 462  & 69.68\% & 152     & 95.0\%  & 267  & 53.83\% & 71      & 93.42\% & 306  & 58.4\%  & 70      & 97.22\% & 171  & 44.76\% & 49      & 87.5\% \\
\midrule
\multicolumn{25}{c}{\textbf{Equivalence}$ ^* $}\\
\midrule
proven by Vampire (\texttt{vampire}, \texttt{casc}, \texttt{casc\_sat} combined)  & 2353 & 100.0\% & 83 & 100.0\% & 264 & 100.0\% & 12 & 100.0\% & 201 & 100.0\% & 8 & 100.0\% & 229 & 100.0\% & 5 & 100.0\% & 218 & 100.0\% & 2 & 100.0\% & 211 & 100.0\% & 7 & 100.0\% \\
\midrule
\multicolumn{25}{c}{\textbf{Non-Equivalence}$ ^{**} $} \\
\bottomrule
    \end{tabular}
    \end{adjustbox}
    \caption{Evaluation for exercise E-2 Murder mystery, formulas E-2-1 -- E-2-5.}\label{table:exercise-dreadsbury}
\end{table*}
\renewcommand{\fiiill}{&&&&&&&&&&&&&&&&&& &&&&}
\begin{table*}
    \scriptsize
    \centering
    \begin{adjustbox}{angle=90}
    \begin{tabular}{l|| rrrr | rrrr | rrrr | rrrr | rrrr | rrrr}
        \toprule 
        Exercise E-2: Murder mystery (E-2-6 - E-2-11) & \multicolumn{4}{c|}{formula 6} & \multicolumn{4}{c|}{formula 7} & \multicolumn{4}{c|}{formula 8} & \multicolumn{4}{c|}{formula 9} & \multicolumn{4}{c|}{formula 10} & \multicolumn{4}{c}{formula 11}\\
        & \defaultColumnHeaders & \defaultColumnHeaders & \defaultColumnHeaders & \defaultColumnHeaders & \defaultColumnHeaders & \defaultColumnHeaders \\
        \midrule
        all attempts                                                                      & 385  & 8.51\%  & 87      & 11.13\% & 475  & 10.5\%  & 92      & 11.76\% & 318  & 7.03\%  & 46      & 5.88\%  & 270  & 5.97\%  & 29      & 3.71\%  & 402   & 8.88\%  & 84       & 10.74\%  & 262   & 5.79\%  & 29       & 3.71\%   \\
equivalent attempts                                                               & 212  & 55.06\% & 11      & 12.64\% & 204  & 42.95\% & 11      & 11.96\% & 205  & 64.47\% & 5       & 10.87\% & 207  & 76.67\% & 5       & 17.24\% & 194   & 48.26\% & 9        & 10.71\%  & 208   & 79.39\% & 8        & 27.59\%  \\
non-equivalent attempts                                                           & 173  & 44.94\% & 76      & 87.36\% & 271  & 57.05\% & 81      & 88.04\% & 113  & 35.53\% & 41      & 89.13\% & 63   & 23.33\% & 24      & 82.76\% & 208   & 51.74\% & 75       & 89.29\%  & 54    & 20.61\% & 21       & 72.41\%  \\
\midrule
\multicolumn{25}{c}{\textbf{Equivalence}$ ^* $}\\
\midrule
proven by Vampire (\texttt{vampire}, \texttt{casc}, \texttt{casc\_sat} combined)  & 212 & 100.0\% & 11 & 100.0\% & 204 & 100.0\% & 11 & 100.0\% & 205 & 100.0\% & 5 & 100.0\% & 207 & 100.0\% & 5 & 100.0\% & 194 & 100.0\% & 9 & 100.0\% & 208 & 100.0\% & 8 & 100.0\% \\
\midrule
\multicolumn{25}{c}{\textbf{Non-Equivalence}$ ^{**} $}\\
\bottomrule
    \end{tabular}
    \end{adjustbox}
    \caption{Evaluation for exercise E-2 Murder mystery, formulas E-2-6 -- E-2-11.}\label{table:exercise-dreadsbury-pt2}
\end{table*}

\renewcommand{\fiiill}{&&&&&&&&&&&&&& &&&&}
\begin{table*}
    \scriptsize
    \centering
    \begin{adjustbox}{angle=90}
    \begin{tabular}{l|| rrrr | rrrr | rrrr | rrrr | rrrr}
        \toprule 
        Exercise E-3: Film Database & \multicolumn{4}{c|}{combined} & \multicolumn{4}{c|}{formula 1} & \multicolumn{4}{c|}{formula 2} & \multicolumn{4}{c}{formula 3} & \multicolumn{4}{c}{formula 4}\\
        & \defaultColumnHeaders & \defaultColumnHeaders & \defaultColumnHeaders & \defaultColumnHeaders & \defaultColumnHeaders \\
        \midrule
        	all attempts & 5938 & 100.0\% & 2227 & 100.0\% & 552 & 9.3\% & 100 & 4.49\% & 1645 & 27.7\% & 283 & 12.71\% & 2429 & 40.91\% & 1180 & 52.99\% & 1312 & 22.09\% & 664 & 29.82\% \\
	equivalent attempts & 1435 & 24.17\% & 275 & 12.35\% & 413 & 74.82\% & 24 & 24.0\% & 386 & 23.47\% & 34 & 12.01\% & 334 & 13.75\% & 142 & 12.03\% & 302 & 23.02\% & 75 & 11.3\% \\
	non-equivalent attempts & 4503 & 75.83\% & 1952 & 87.65\% & 139 & 25.18\% & 76 & 76.0\% & 1259 & 76.53\% & 249 & 87.99\% & 2095 & 86.25\% & 1038 & 87.97\% & 1010 & 76.98\% & 589 & 88.7\% \\
    \midrule
\multicolumn{21}{c}{\textbf{Equivalence}$ ^* $}\\
\midrule
	proven by Vampire (\texttt{vampire}, \texttt{casc}, \texttt{casc\_sat} combined) & 1435 & 100.0\% & 275 & 100.0\% & 413 & 100.0\% & 24 & 100.0\% & 386 & 100.0\% & 34 & 100.0\% & 334 & 100.0\% & 142 & 100.0\% & 302 & 100.0\% & 75 & 100.0\% \\
    \midrule
\multicolumn{21}{c}{\textbf{Non-Equivalence}$ ^{**} $}\\
\bottomrule
    \end{tabular}
    \end{adjustbox}
    \caption{Evaluation for exercise E-3 Film Database.}\label{table:exercise-film-database}
\end{table*}

\renewcommand{\fiiill}{&&&&&&&&&&&&&& &&&&}
\begin{table*}
    \scriptsize
    \centering
    \begin{adjustbox}{angle=90}
    \begin{tabular}{l|| rrrr | rrrr | rrrr | rrrr | rrrr}
        \toprule
        Exercise E-4: Social Network & \multicolumn{4}{c|}{combined} & \multicolumn{4}{c|}{formula 1} & \multicolumn{4}{c|}{formula 2} & \multicolumn{4}{c}{formula 3} & \multicolumn{4}{c}{formula 4}\\
        & \defaultColumnHeaders & \defaultColumnHeaders & \defaultColumnHeaders & \defaultColumnHeaders & \defaultColumnHeaders \\
        \midrule
        all attempts & 6119 & 100.0\% & 1711 & 100.0\% & 717 & 11.72\% & 47 & 2.75\% & 864 & 14.12\% & 173 & 10.11\% & 3545 & 57.93\% & 1075 & 62.83\% & 993 & 16.23\% & 416 & 24.31\% \\
        equivalent attempts & 1615 & 26.39\% & 133 & 7.77\% & 464 & 64.71\% & 9 & 19.15\% & 436 & 50.46\% & 8 & 4.62\% & 359 & 10.13\% & 64 & 5.95\% & 356 & 35.85\% & 52 & 12.5\% \\
        non-equivalent attempts & 4504 & 73.61\% & 1578 & 92.23\% & 253 & 35.29\% & 38 & 80.85\% & 428 & 49.54\% & 165 & 95.38\% & 3186 & 89.87\% & 1011 & 94.05\% & 637 & 64.15\% & 364 & 87.5\% \\
        \midrule
        \multicolumn{21}{c}{\textbf{Equivalence}$ ^* $}\\
        \midrule
        proven by Vampire (\texttt{vampire}, \texttt{casc}, \texttt{casc\_sat} combined) & 1615 & 100.0\% & 133 & 100.0\% & 464 & 100.0\% & 9 & 100.0\% & 436 & 100.0\% & 8 & 100.0\% & 359 & 100.0\% & 64 & 100.0\% & 356 & 100.0\% & 52 & 100.0\% \\
        \midrule
        \multicolumn{21}{c}{\textbf{Non-Equivalence}$ ^{**} $}\\
\bottomrule
    \end{tabular}
    \end{adjustbox}
    \caption{Evaluation for exercise E-4 Social Network.}\label{table:exercise-social-network}
\end{table*}

\renewcommand{\fiiill}{&&&&&&&&&&&&&&&&&& &&&&}

\begin{table*}
    \scriptsize
    \centering
\scalebox{0.97}{    \begin{adjustbox}{angle=90}
    \begin{tabular}{l|| rrrr | rrrr | rrrr | rrrr | rrrr | rrrr}
        \toprule
        Exercise E-5: Volleyball & \multicolumn{4}{c|}{combined} & \multicolumn{4}{c|}{formula 1} & \multicolumn{4}{c|}{formula 2} & \multicolumn{4}{c|}{formula 3} & \multicolumn{4}{c|}{formula 4} & \multicolumn{4}{c}{formula 5}\\
        & \defaultColumnHeaders & \defaultColumnHeaders & \defaultColumnHeaders & \defaultColumnHeaders & \defaultColumnHeaders & \defaultColumnHeaders \\
        \midrule
        all attempts & 6370 & 100.0\% & 3005 & 100.0\% & 417 & 6.55\% & 31 & 1.03\% & 2686 & 42.17\% & 973 & 32.38\% & 978 & 15.35\% & 506 & 16.84\% & 1541 & 24.19\% & 1175 & 39.1\% & 748 & 11.74\% & 320 & 10.65\% \\
        equivalent attempts & 1424 & 22.35\% & 393 & 13.08\% & 391 & 93.76\% & 14 & 45.16\% & 290 & 10.8\% & 71 & 7.3\% & 311 & 31.8\% & 76 & 15.02\% & 211 & 13.69\% & 160 & 13.62\% & 221 & 29.55\% & 72 & 22.5\% \\
        non-equivalent attempts & 4946 & 77.65\% & 2612 & 86.92\% & 26 & 6.24\% & 17 & 54.84\% & 2396 & 89.2\% & 902 & 92.7\% & 667 & 68.2\% & 430 & 84.98\% & 1330 & 86.31\% & 1015 & 86.38\% & 527 & 70.45\% & 248 & 77.5\% \\
        \midrule
        \multicolumn{25}{c}{\textbf{Equivalence}$ ^* $}\\
        \midrule
        proven by Vampire (\texttt{vampire}, \texttt{casc}, \texttt{casc\_sat} combined) & 1424 & 100.0\% & 393 & 100.0\% & 391 & 100.0\% & 14 & 100.0\% & 290 & 100.0\% & 71 & 100.0\% & 311 & 100.0\% & 76 & 100.0\% & 211 & 100.0\% & 160 & 100.0\% & 221 & 100.0\% & 72 & 100.0\% \\
        \midrule
        \multicolumn{25}{c}{\textbf{Non-Equivalence}$ ^{**} $}\\
\bottomrule
    \end{tabular}
    \end{adjustbox}}
    \caption{Evaluation for exercise E-5 Volleyball.}\label{table:exercise-volleyball}
\end{table*}

\renewcommand{\fiiill}{&&&&&&&&&&&&&&&&&& &&&&}
\begin{table*}
    \tiny
    \centering
    \begin{adjustbox}{angle=90}
    \begin{tabular}{l|| rrrr | rrrr | rrrr | rrrr | rrrr | rrrr}
        \toprule
        Exercise E-6: Graph & \multicolumn{4}{c|}{combined} & \multicolumn{4}{c|}{formula 1} & \multicolumn{4}{c|}{formula 2} & \multicolumn{4}{c|}{formula 3} & \multicolumn{4}{c|}{formula 4} & \multicolumn{4}{c}{formula 5}\\
        & \defaultColumnHeaders & \defaultColumnHeaders & \defaultColumnHeaders & \defaultColumnHeaders & \defaultColumnHeaders & \defaultColumnHeaders \\
        \midrule
        all attempts & 40362 & 100.0\% & 5844 & 100.0\% & 4618 & 11.44\% & 513 & 8.78\% & 9048 & 22.42\% & 1096 & 18.75\% & 9571 & 23.71\% & 1070 & 18.31\% & 4760 & 11.79\% & 971 & 16.62\% & 12365 & 30.64\% & 2194 & 37.54\% \\
        equivalent attempts & 7235 & 17.93\% & 265 & 4.53\% & 1836 & 39.76\% & 52 & 10.14\% & 1545 & 17.08\% & 26 & 2.37\% & 1370 & 14.31\% & 25 & 2.34\% & 1281 & 26.91\% & 105 & 10.81\% & 1203 & 9.73\% & 57 & 2.6\% \\
        non-equivalent attempts & 33127 & 82.07\% & 5579 & 95.47\% & 2782 & 60.24\% & 461 & 89.86\% & 7503 & 82.92\% & 1070 & 97.63\% & 8201 & 85.69\% & 1045 & 97.66\% & 3479 & 73.09\% & 866 & 89.19\% & 11162 & 90.27\% & 2137 & 97.4\% \\
        \midrule
        \multicolumn{25}{c}{\textbf{Equivalence}$ ^* $}\\
        \midrule
        proven by Vampire (\texttt{vampire}, \texttt{casc}, \texttt{casc\_sat} combined) & 7235 & 100.0\% & 265 & 100.0\% & 1836 & 100.0\% & 52 & 100.0\% & 1545 & 100.0\% & 26 & 100.0\% & 1370 & 100.0\% & 25 & 100.0\% & 1281 & 100.0\% & 105 & 100.0\% & 1203 & 100.0\% & 57 & 100.0\% \\
        \midrule
        \multicolumn{25}{c}{\textbf{Non-Equivalence}$ ^{**} $}\\
\bottomrule
    \end{tabular}
    \end{adjustbox}
    \caption{Evaluation for exercise E-6 Graph.}\label{table:exercise-graph-ii}
\end{table*}

\renewcommand{\fiiill}{&&&&&&&&&&&&&&&&}
\begin{table*}
    \footnotesize
    \centering
    \begin{adjustbox}{angle=90}
    \begin{tabular}{l|| rrrr | rrrr | rrrr | rrrr}
        \toprule
        Exercise E-7: Murder mystery II & \multicolumn{4}{c|}{combined} & \multicolumn{4}{c|}{formula 1} & \multicolumn{4}{c|}{formula 2} & \multicolumn{4}{c}{formula 3} \\
        & \defaultColumnHeaders & \defaultColumnHeaders & \defaultColumnHeaders  & \defaultColumnHeaders \\
        \midrule
        all attempts & 987 & 100.0\% & 215 & 100.0\% & 497 & 50.35\% & 105 & 48.84\% & 268 & 27.15\% & 78 & 36.28\% & 222 & 22.49\% & 32 & 14.88\% \\
        equivalent attempts & 294 & 29.79\% & 16 & 7.44\% & 96 & 19.32\% & 5 & 4.76\% & 93 & 34.7\% & 5 & 6.41\% & 105 & 47.3\% & 6 & 18.75\% \\
        non-equivalent attempts & 693 & 70.21\% & 199 & 92.56\% & 401 & 80.68\% & 100 & 95.24\% & 175 & 65.3\% & 73 & 93.59\% & 117 & 52.7\% & 26 & 81.25\% \\
        \midrule
        \multicolumn{17}{c}{\textbf{Equivalence}$ ^* $}\\
        \midrule
        proven by Vampire (\texttt{vampire}, \texttt{casc}, \texttt{casc\_sat} combined) & 294 & 100.0\% & 16 & 100.0\% & 96 & 100.0\% & 5 & 100.0\% & 93 & 100.0\% & 5 & 100.0\% & 105 & 100.0\% & 6 & 100.0\% \\
        \midrule
        \multicolumn{17}{c}{\textbf{Non-Equivalence}$ ^{**} $}\\
\bottomrule
    \end{tabular}
    \end{adjustbox}
    \caption{Evaluation for exercise E-7 Murder mystery II.}\label{table:exercise-dreadsbury-ii}
\end{table*}

\renewcommand{\fiiill}{&&&&&&&&&&&&&&&&}
\begin{table*}
    \footnotesize
    \centering
    \begin{adjustbox}{angle=90}
    \begin{tabular}{l|| rrrr | rrrr | rrrr | rrrr}
        \toprule
        Exercise E-8: Strings & \multicolumn{4}{c|}{combined} & \multicolumn{4}{c|}{formula 1} & \multicolumn{4}{c|}{formula 2} & \multicolumn{4}{c}{formula 3} \\
        & \defaultColumnHeaders & \defaultColumnHeaders & \defaultColumnHeaders  & \defaultColumnHeaders \\
        \midrule
        all attempts & 640 & 100.0\% & 289 & 100.0\% & 160 & 25.0\% & 74 & 25.61\% & 175 & 27.34\% & 102 & 35.29\% & 305 & 47.66\% & 113 & 39.1\% \\
        equivalent attempts & 87 & 13.59\% & 41 & 14.19\% & 41 & 25.63\% & 15 & 20.27\% & 32 & 18.29\% & 18 & 17.65\% & 14 & 4.59\% & 8 & 7.08\% \\
        non-equivalent attempts & 553 & 86.41\% & 248 & 85.81\% & 119 & 74.38\% & 59 & 79.73\% & 143 & 81.71\% & 84 & 82.35\% & 291 & 95.41\% & 105 & 92.92\% \\
        \midrule
        \multicolumn{17}{c}{\textbf{Equivalence}$ ^* $}\\
        \midrule
        proven by Vampire (\texttt{vampire}, \texttt{casc}, \texttt{casc\_sat} combined) & 87 & 100.0\% & 41 & 100.0\% & 41 & 100.0\% & 15 & 100.0\% & 32 & 100.0\% & 18 & 100.0\% & 14 & 100.0\% & 8 & 100.0\% \\
        \midrule
        \multicolumn{17}{c}{\textbf{Non-Equivalence}$ ^{**} $}\\
\bottomrule
    \end{tabular}
    \end{adjustbox}
    \caption{Evaluation for exercise E-8 Strings.}\label{table:exercise-strings}
\end{table*}

\renewcommand{\fiiill}{&&&&&&&&&&&&&&&&&&&&}
\begin{table*}
    \scriptsize
    \centering
    \begin{adjustbox}{angle=90}
    \begin{tabular}{l|| rrrr | rrrr | rrrr | rrrr | rrrr}
        \toprule 
        Exercise E-9: Theorems & \multicolumn{4}{c|}{combined} & \multicolumn{4}{c|}{formula 1} & \multicolumn{4}{c|}{formula 2} & \multicolumn{4}{c|}{formula 3} & \multicolumn{4}{c}{formula 4}  \\
        & \defaultColumnHeaders & \defaultColumnHeaders & \defaultColumnHeaders  & \defaultColumnHeaders & \defaultColumnHeaders \\
        \midrule
        all attempts & 488 & 100.0\% & 114 & 100.0\% & 43 & 8.81\% & 5 & 4.39\% & 146 & 29.92\% & 38 & 33.33\% & 193 & 39.55\% & 40 & 35.09\% & 106 & 21.72\% & 31 & 27.19\% \\
	equivalent attempts & 176 & 36.07\% & 15 & 13.16\% & 39 & 90.7\% & 2 & 40.0\% & 50 & 34.25\% & 7 & 18.42\% & 47 & 24.35\% & 3 & 7.5\% & 40 & 37.74\% & 3 & 9.68\% \\
	non-equivalent attempts & 312 & 63.93\% & 99 & 86.84\% & 4 & 9.3\% & 3 & 60.0\% & 96 & 65.75\% & 31 & 81.58\% & 146 & 75.65\% & 37 & 92.5\% & 66 & 62.26\% & 28 & 90.32\% \\
    \midrule
\multicolumn{21}{c}{\textbf{Equivalence}$ ^* $}\\
\midrule
	proven by Vampire (\texttt{vampire}, \texttt{casc}, \texttt{casc\_sat} combined) & 176 & 100.0\% & 15 & 100.0\% & 39 & 100.0\% & 2 & 100.0\% & 50 & 100.0\% & 7 & 100.0\% & 47 & 100.0\% & 3 & 100.0\% & 40 & 100.0\% & 3 & 100.0\% \\
    \midrule
\multicolumn{21}{c}{\textbf{Non-Equivalence}$ ^{**} $}\\
\bottomrule
    \end{tabular}
    \end{adjustbox}
    \caption{Evaluation of exercise E-9 Theorems.}\label{table:exercise-theorems}
\end{table*}

\renewcommand{\fiiill}{&&&&&&&&&&&&&&&&&&&& &&&&}
\begin{table*}
    \tiny
    \centering
    \begin{adjustbox}{angle=90}
    \begin{tabular}{l|| rrrr | rrrr | rrrr | rrrr | rrrr | rrrr}
        \toprule 
        Exercise E-10: Millisoft & \multicolumn{4}{c|}{combined} & \multicolumn{4}{c|}{formula 1} & \multicolumn{4}{c|}{formula 2} & \multicolumn{4}{c|}{formula 3} & \multicolumn{4}{c|}{formula 4} & \multicolumn{4}{c}{formula 5}  \\
        & \defaultColumnHeaders & \defaultColumnHeaders & \defaultColumnHeaders  & \defaultColumnHeaders & \defaultColumnHeaders & \defaultColumnHeaders \\
        \midrule
        all attempts & 28095 & 100.0\% & 9336 & 100.0\% & 1888 & 6.72\% & 126 & 1.35\% & 2694 & 9.59\% & 483 & 5.17\% & 7170 & 25.52\% & 2949 & 31.59\% & 3906 & 13.9\% & 981 & 10.51\% & 12437 & 44.27\% & 4797 & 51.38\% \\
	equivalent attempts & 5220 & 18.58\% & 774 & 8.29\% & 1319 & 69.86\% & 19 & 15.08\% & 1206 & 44.77\% & 61 & 12.63\% & 933 & 13.01\% & 311 & 10.55\% & 990 & 25.35\% & 53 & 5.4\% & 772 & 6.21\% & 330 & 6.88\% \\
	non-equivalent attempts & 22875 & 81.42\% & 8562 & 91.71\% & 569 & 30.14\% & 107 & 84.92\% & 1488 & 55.23\% & 422 & 87.37\% & 6237 & 86.99\% & 2638 & 89.45\% & 2916 & 74.65\% & 928 & 94.6\% & 11665 & 93.79\% & 4467 & 93.12\% \\
    \midrule
\multicolumn{25}{c}{\textbf{Equivalence}$ ^* $}\\
\midrule
	proven by Vampire (\texttt{vampire}, \texttt{casc}, \texttt{casc\_sat} combined) & 5220 & 100.0\% & 774 & 100.0\% & 1319 & 100.0\% & 19 & 100.0\% & 1206 & 100.0\% & 61 & 100.0\% & 933 & 100.0\% & 311 & 100.0\% & 990 & 100.0\% & 53 & 100.0\% & 772 & 100.0\% & 330 & 100.0\% \\
    \midrule
\multicolumn{25}{c}{\textbf{Non-Equivalence}$ ^{**} $}\\
\bottomrule
    \end{tabular}
\end{adjustbox}
    \caption{Evaluation for exercise E-10 Millisoft.}\label{table:exercise-millisoft}
\end{table*}

\renewcommand{\fiiill}{&&&&&&&&&&&&&&&&&&&&}
\begin{table*}
    \scriptsize
    \centering
    \begin{adjustbox}{angle=90}
    \begin{tabular}{l|| rrrr | rrrr | rrrr | rrrr | rrrr }
        \toprule 
        Exercise E-11: Web pages & \multicolumn{4}{c|}{combined} & \multicolumn{4}{c|}{formula 1} & \multicolumn{4}{c|}{formula 2} & \multicolumn{4}{c|}{formula 3} & \multicolumn{4}{c}{formula 4} \\
        & \defaultColumnHeaders & \defaultColumnHeaders & \defaultColumnHeaders  & \defaultColumnHeaders & \defaultColumnHeaders \\
        \midrule
        all attempts & 11543 & 100.0\% & 3771 & 100.0\% & 4604 & 39.89\% & 1261 & 33.44\% & 2148 & 18.61\% & 845 & 22.41\% & 2945 & 25.51\% & 1198 & 31.77\% & 1846 & 15.99\% & 467 & 12.38\% \\
	equivalent attempts & 2784 & 24.12\% & 170 & 4.51\% & 752 & 16.33\% & 40 & 3.17\% & 691 & 32.17\% & 43 & 5.09\% & 659 & 22.38\% & 65 & 5.43\% & 682 & 36.94\% & 22 & 4.71\% \\
	non-equivalent attempts & 8759 & 75.88\% & 3601 & 95.49\% & 3852 & 83.67\% & 1221 & 96.83\% & 1457 & 67.83\% & 802 & 94.91\% & 2286 & 77.62\% & 1133 & 94.57\% & 1164 & 63.06\% & 445 & 95.29\% \\
    \midrule
\multicolumn{21}{c}{\textbf{Equivalence}$ ^* $}\\
\midrule
	proven by Vampire (\texttt{vampire}, \texttt{casc}, \texttt{casc\_sat} combined) & 2784 & 100.0\% & 170 & 100.0\% & 752 & 100.0\% & 40 & 100.0\% & 691 & 100.0\% & 43 & 100.0\% & 659 & 100.0\% & 65 & 100.0\% & 682 & 100.0\% & 22 & 100.0\% \\
    \midrule
\multicolumn{21}{c}{\textbf{Non-Equivalence}$ ^{**} $}\\
\bottomrule
    \end{tabular}
    \end{adjustbox}
    \caption{Evaluation of exercise E-11 Web pages.}\label{table:exercise-web-pages}
\end{table*}

\renewcommand{\fiiill}{&&&&&&&&&&&&&&&&&&&&}
\begin{table*}
    \scriptsize
    \centering
    \begin{adjustbox}{angle=90}
    \begin{tabular}{l|| rrrr | rrrr | rrrr | rrrr | rrrr }
        \toprule 
        Exercise E-12: Recipe & \multicolumn{4}{c|}{combined} & \multicolumn{4}{c|}{formula 1} & \multicolumn{4}{c|}{formula 2} & \multicolumn{4}{c|}{formula 3} & \multicolumn{4}{c}{formula 4} \\
        & \defaultColumnHeaders & \defaultColumnHeaders & \defaultColumnHeaders  & \defaultColumnHeaders & \defaultColumnHeaders \\
        \midrule
        all attempts & 11200 & 100.0\% & 3778 & 100.0\% & 1756 & 15.68\% & 165 & 4.37\% & 2431 & 21.71\% & 740 & 19.59\% & 3819 & 34.1\% & 1735 & 45.92\% & 3194 & 28.52\% & 1138 & 30.12\% \\
	equivalent attempts & 3388 & 30.25\% & 292 & 7.73\% & 1024 & 58.31\% & 19 & 11.52\% & 947 & 38.96\% & 71 & 9.59\% & 803 & 21.03\% & 150 & 8.65\% & 614 & 19.22\% & 52 & 4.57\% \\
	non-equivalent attempts & 7812 & 69.75\% & 3486 & 92.27\% & 732 & 41.69\% & 146 & 88.48\% & 1484 & 61.04\% & 669 & 90.41\% & 3016 & 78.97\% & 1585 & 91.35\% & 2580 & 80.78\% & 1086 & 95.43\% \\
    \midrule
\multicolumn{21}{c}{\textbf{Equivalence}$ ^* $}\\
\midrule
	proven by Vampire (\texttt{vampire}, \texttt{casc}, \texttt{casc\_sat} combined) & 3388 & 100.0\% & 292 & 100.0\% & 1024 & 100.0\% & 19 & 100.0\% & 947 & 100.0\% & 71 & 100.0\% & 803 & 100.0\% & 150 & 100.0\% & 614 & 100.0\% & 52 & 100.0\% \\
    \midrule
\multicolumn{21}{c}{\textbf{Non-Equivalence}$ ^{**} $}\\
\bottomrule
    \end{tabular}
    \end{adjustbox}
    \caption{Evaluation of exercise E-12 Recipe.}\label{table:exercise-recipe}
\end{table*}

\renewcommand{\fiiill}{&&&&&&&&&&&&&&&&&&&&}
\begin{table*}
    \scriptsize
    \centering
    \begin{adjustbox}{angle=90}
    \begin{tabular}{l|| rrrr | rrrr | rrrr | rrrr | rrrr }
        \toprule 
        Exercise E-13: Book Collection & \multicolumn{4}{c|}{combined} & \multicolumn{4}{c|}{formula 1} & \multicolumn{4}{c|}{formula 2} & \multicolumn{4}{c|}{formula 3} & \multicolumn{4}{c}{formula 4} \\
        & \defaultColumnHeaders & \defaultColumnHeaders & \defaultColumnHeaders  & \defaultColumnHeaders & \defaultColumnHeaders \\
        \midrule
        all attempts & 2040 & 100.0\% & 1230 & 100.0\% & 267 & 13.09\% & 126 & 10.24\% & 582 & 28.53\% & 355 & 28.86\% & 582 & 28.53\% & 365 & 29.67\% & 609 & 29.85\% & 384 & 31.22\% \\
	equivalent attempts & 345 & 16.91\% & 112 & 9.11\% & 118 & 44.19\% & 24 & 19.05\% & 92 & 15.81\% & 27 & 7.61\% & 86 & 14.78\% & 40 & 10.96\% & 49 & 8.05\% & 21 & 5.47\% \\
	non-equivalent attempts & 1695 & 83.09\% & 1118 & 90.89\% & 149 & 55.81\% & 102 & 80.95\% & 490 & 84.19\% & 328 & 92.39\% & 496 & 85.22\% & 325 & 89.04\% & 560 & 91.95\% & 363 & 94.53\% \\
    \midrule
\multicolumn{21}{c}{\textbf{Equivalence}$ ^* $}\\
\midrule
	proven by Vampire (\texttt{vampire}, \texttt{casc}, \texttt{casc\_sat} combined) & 345 & 100.0\% & 112 & 100.0\% & 118 & 100.0\% & 24 & 100.0\% & 92 & 100.0\% & 27 & 100.0\% & 86 & 100.0\% & 40 & 100.0\% & 49 & 100.0\% & 21 & 100.0\% \\
    \midrule
\multicolumn{21}{c}{\textbf{Non-Equivalence}$ ^{**} $}\\
\bottomrule
    \end{tabular}
    \end{adjustbox}
    \caption{Evaluation of exercise E-13 Book Collection.}\label{table:exercise-book-collection}
\end{table*}

\renewcommand{\fiiill}{&&&&&&&&&&&&&&&&}
\begin{table*}
    \footnotesize
    \centering
    \begin{adjustbox}{angle=90}
    \begin{tabular}{l|| rrrr | rrrr | rrrr | rrrr }
        \toprule
        Exercise E-14: Faculty Conference & \multicolumn{4}{c|}{combined} & \multicolumn{4}{c|}{formula 1} & \multicolumn{4}{c|}{formula 2} & \multicolumn{4}{c}{formula 3} \\
        & \defaultColumnHeaders & \defaultColumnHeaders & \defaultColumnHeaders  & \defaultColumnHeaders  \\
        \midrule
        all attempts & 917 & 100.0\% & 622 & 100.0\% & 607 & 66.19\% & 379 & 60.93\% & 235 & 25.63\% & 186 & 29.9\% & 75 & 8.18\% & 57 & 9.16\% \\
        equivalent attempts & 77 & 8.4\% & 36 & 5.79\% & 44 & 7.25\% & 15 & 3.96\% & 16 & 6.81\% & 12 & 6.45\% & 17 & 22.67\% & 9 & 15.79\% \\
        non-equivalent attempts & 840 & 91.6\% & 586 & 94.21\% & 563 & 92.75\% & 364 & 96.04\% & 219 & 93.19\% & 174 & 93.55\% & 58 & 77.33\% & 48 & 84.21\% \\
        \midrule
        \multicolumn{17}{c}{\textbf{Equivalence}$ ^* $}\\
        \midrule
        proven by Vampire (\texttt{vampire}, \texttt{casc}, \texttt{casc\_sat} combined) & 77 & 100.0\% & 36 & 100.0\% & 44 & 100.0\% & 15 & 100.0\% & 16 & 100.0\% & 12 & 100.0\% & 17 & 100.0\% & 9 & 100.0\% \\
        \midrule
        \multicolumn{17}{c}{\textbf{Non-Equivalence}$ ^{**} $}\\
\bottomrule
    \end{tabular}
    \end{adjustbox}
    \caption{Evaluation of exercise E-14 Faculty Conference.}\label{table:exercise-faculty-conference}
\end{table*}
 
\fi
\end{document}